\documentclass[english,11pt]{article}
\usepackage[T1]{fontenc}
\usepackage[latin1]{inputenc}
\usepackage{geometry}
\usepackage{float}
\usepackage{mathrsfs}
\usepackage{amsmath}
\usepackage{amssymb}
\usepackage{graphicx}
\usepackage{hyperref}
\usepackage[color=yellow]{todonotes}
\hypersetup{
    colorlinks=true,
    linkcolor=blue,
    filecolor=magenta,      
    urlcolor=cyan,
    citecolor = red,
}

\makeatletter

\floatstyle{ruled}
\newfloat{algorithm}{tbp}{loa}
\providecommand{\algorithmname}{Algorithm}
\floatname{algorithm}{\protect\algorithmname}

\usepackage{amsthm}

\usepackage{mathrsfs}

\usepackage{amsfonts}

\usepackage{epsfig}

\usepackage{bm}

\usepackage{mathrsfs}

\usepackage{enumerate}

\@ifundefined{definecolor}{\@ifundefined{definecolor}
 {\@ifundefined{definecolor}
 {\usepackage{color}}{}
}{}
}{}

\usepackage{subfig}\usepackage[all]{xy}

\newtheorem{theorem}{Theorem}[section]
\newtheorem{lem}{Lemma}[section]
\newtheorem{rem}{Remark}[section]
\newtheorem{prop}{Proposition}[section]
\newtheorem{cor}{Corollary}[section]
\newcounter{hypA}

\newcounter{hypB}

\newcounter{hypD}

\usepackage{babel}\date{}

\def\ricc{\mbox{\rm Ricc}}

\newcommand{\EE}{\mathbb{E}}

\newcommand{\LL}{\mathbb{L}}

\newcommand{\Ea}{ {\cal E }}

\newcommand{\Fa}{ {\cal F }}

\newcommand{\Pa}{ {\cal P }}

\newcommand{\point}{\mbox{\LARGE .}}

\def \SS{\mathbb{S}}
\def \EE{\mathbb{E}}

\def \LL{\mathbb{L}}

\makeatother

\begin{document}

\begin{center}

{\Large \textbf{Log-Normalization Constant
Estimation using the Ensemble Kalman-Bucy Filter with Application to High-Dimensional Models}}

\vspace{0.5cm}

BY DAN CRISAN$^{1}$, PIERRE DEL MORAL$^{2}$, AJAY JASRA$^{3}$ \& HAMZA RUZAYQAT$^{3}$

{\footnotesize $^{1}$Department of Mathematics, Imperial College London, London, SW7 2AZ, UK.}
{\footnotesize E-Mail:\,} \texttt{\emph{\footnotesize d.crisan@ic.ac.uk}}\\
{\footnotesize $^{2}$Center INRIA Bordeaux Sud-Ouest \& Institut de Mathematiques de Bordeaux, Bordeaux, 33405, FR.}
{\footnotesize E-Mail:\,} \texttt{\emph{\footnotesize pierre.del-moral@inria.fr}}\\
{\footnotesize $^{3}$Computer, Electrical and Mathematical Sciences and Engineering Division, King Abdullah University of Science and Technology, Thuwal, 23955, KSA.}
{\footnotesize E-Mail:\,} \texttt{\emph{\footnotesize ajay.jasra@kaust.edu.sa, hamza.ruzayqat@kaust.edu.sa}}
\end{center}

\begin{abstract}
In this article we consider the estimation of the log-normalization constant associated to a
class of continuous-time filtering models. In particular, we consider ensemble Kalman-Bucy
filter based estimates based upon several nonlinear Kalman-Bucy diffusions. Based upon new
conditional bias results for the mean of the afore-mentioned methods, we analyze the empirical log-scale normalization constants
in terms of their $\mathbb{L}_n-$errors and conditional bias. Depending on the type of nonlinear Kalman-Bucy diffusion, 
we show that these are of order $(t^{1/2}/N^{1/2}) + t/N$ or $1/N^{1/2}$ ($\mathbb{L}_n-$errors) and of order $[t+t^{1/2}]/N$ or $1/N$ (conditional bias), where $t$ is the time horizon and $N$ is the ensemble size.
Finally, we use these results for online static parameter estimation for above
filtering models and implement the methodology for both linear and nonlinear models.\\[1mm]
\end{abstract}

\section{Introduction}

The filtering problem concerns the recursive estimation of a partially observed Markov process conditioned on a path of observations. It is found
in a wide class of real applications, including finance, applied mathematics and engineering; see \cite{BC09} for instance. This article focusses upon the computation of
the normalizing constant for certain classes of filtering problems, to be introduced below. These normalizing constants can be of interest in statistics and engineering
for model selection and/or parameter estimation.

In this article we study linear and Gaussian models in continuous-time. It is well-known that, for such models,  under some minimal assumptions, the filter is Gaussian,
with mean satisfying the Kalman-Bucy (KB) equations and the covariance matrix obeying a Riccati equation. In many cases of practical interest, these latter equations
may not be computable, for instance if one does not have access to an entire trajectory of data. Instead one can use some suitable approximations obtained via time-discretization methods. However, even then, if the dimension $r_1$ of the state-vector is very large, the numerical approximation of the KB and Riccati equations can become computationally prohibitive, often with a computational effort of at least $\mathcal{O}(r_1^2)$ per update.
The case where $r_1$ is large occurs in many applications including ocean and atmosphere science (e.g.~\cite{atmos}) and oil reservoir simulations (e.g.~\cite{oil}).
A rather elegant and successful solution, in discrete-time, to this problem was developed in \cite{enkf} in the guise of the ensemble Kalman filter (EnKF), which can
reduce the cost to $\mathcal{O}(r_1)$.

The EnKF can be interpreted as a mean-field particle approximation of a conditional McKean-Vlasov diffusion (the K-B diffusion). This latter diffusion shares the same law
as the filter associated to the linear and Gaussian model in continuous-time. Hence a possible alternative to recursively solving the K-B and Riccati equations
is to generate $N\in\mathbb{N}$ independent copies from the K-B diffusion and use a simple Monte Carlo estimate for expectations with respect to~the filter. However, the diffusion process
cannot be simulated exactly, but can be approximated in a principled way by allowing the $N$ samples to interact; precise details are given in Section \ref{desc-sec-intro}.
The resulting method, named the Ensemble Kalman-Bucy Filter (EnKBF), is by now rather well-understood with several contributions on its convergence (as $N\rightarrow\infty$) analysis; see for instance \cite{BD20,BD17,dt-2016,legland_enkf,tong}.

In this work we focus upon using several versions of EnKBF for an online estimate of the normalization constant. In particular the contributions of this work are:
\begin{enumerate}
\item{New results on the conditional bias of the mean using the EnKBF}
\item{A derivation of an estimate of the normalization constant using the EnKBF.}
\item{A proof that the $\mathbb{L}_n-$error of the estimate on the log-scale is of $\mathcal{O}\big(\sqrt{\frac{t}{N}}+\frac{t}{N}\big)$ or $\mathcal{O}\big(\tfrac{1}{\sqrt{N}}\big)$, depending on the nonlinear Kalman-Bucy diffusion and where $t$ is the time parameter.}
\item{A proof that the conditional bias of the estimate on the log-scale is of $\mathcal{O}\big(\tfrac{t+\sqrt{t}}{N}\big)$ or $\mathcal{O}\big(\frac{1}{N}\big)$, depending on the nonlinear Kalman-Bucy diffusion.}
\item{A development of a method that uses this estimate to perform online static parameter estimation.}
\end{enumerate}
The result in 1.~is of independent interest, but is used directly in 4..
To the best of our knowledge the estimate in 2.~is new and can be computed with a little extra computational cost over applying
an EnKBF (under time discretization). Whilst several authors have used the EnKF for normalization constant estimation, e.g.~\cite{enmcmc}, we
have not seen the EnKBF version investigated in the literature. 


In addition to contribution 3.~\& 4., the results establish the decay of the mean square error (MSE), for instance, of the log-normalizing constant estimate as the time
parameter increases. This rate is expected to occur in practice, as we will see in simulations, and parallels other results found in the literature for particle
estimates of the normalization constant (e.g.~such as particle filters \cite{cerou}). In relation to 5.,~if one assumes that the model of interest has several unknown and time-homogeneous
parameters, $\theta$ say, then one is interested to estimate such parameters, for instance using likelihood based methods. We show how our estimate can be leveraged for this latter task. In this paper we use the simultaneous stochastic perturbation stochastic approximation (SPSA) method \cite{spall} popular in engineering applications. SPSA is based upon a finite difference estimator of the gradient, w.r.t.~$\theta$, of the log-normalizing constant. It constructs a stochastic approximation scheme for the estimation of static parameters. We are not aware of  this method being used for the EnKBF.
As it is a zero-order optimization method, we expect to be computationally less expensive than resorting to using other estimates of the gradient (of the log-normalizing constant).
It should be noted that our work focusses upon the standard or `vanilla' EnKBF, the deterministic EnKBF \cite{sakov} and deterministic transport type EnKBFs \cite{reich};
other extensions are possible, for instance, to the feedback particle filter.

This article is structured as follows. In Section \ref{desc-sec-intro} we provide details on the class of filtering problems that we address as well
as details on the ensemble Kalman Bucy filter. 
The conditional bias of the mean associated to various EnKBFs is also presented there.
In Section \ref{sec:nc} we discuss how the normalizing constant estimate can be computed, as well as its $\mathbb{L}_n-$error and conditional bias.
The latter is supported by numerical results checking that the order of convergence rate indeed holds even under naive Euler discretization.
In Section \ref{sec:static_par} we show how the normalizing constant estimate can be used in online static parameter estimation problems.


\section{Description of the Model and Algorithms}\label{desc-sec-intro}

\subsection{The Kalman-Bucy Filter}

Let $(\Omega ,\mathcal{F},\mathbb{P})$ 
be a probability space together
with a filtration $( \mathcal{G}_t) _{t\ge 0}$ 
which satisfies the usual conditions. On $(\Omega ,\mathcal{F},\mathbb{P})$
we consider two $\mathcal{G}_t$-adapted processes $X=\{X_t,\ t\ge 0\}$ and $Y=\{Y_t,\ t\ge 0\}$ that form a time homogeneous linear-Gaussian filtering model of the following form
\begin{equation}\label{lin-Gaussian-diffusion-filtering}
\left\{
\begin{array}{rcl}
dX_t&=&A~X_t~dt~+~R^{1/2}_{1}~dW_t\\
dY_t&=&C~X_t~dt~+~R^{1/2}_{2}~dV_{t}.
\end{array}
\right.
\end{equation}
In the above display, $(W_t,V_t)$ is an $\mathcal{G}_t$-adapted $(r_1+r_2)$-dimensional Brownian motion, $X_0$ is an $\mathcal{G}_0$-measurable $r_1$-valued
Gaussian random vector with mean and covariance matrix $(\EE(X_0),P_0)$ (independent of $(W_t,V_t)$), the symmetric matrices $R^{1/2}_{1}$ and $R^{1/2}_{2}$ are invertible, $A$ is a square  $(r_1\times r_1)$-matrix, $C$ is an  $(r_2\times r_1)$-matrix, and $Y_0=0$.  We let $\Fa_t=\sigma\left(Y_s,~s\leq t\right)$ be the filtration generated by the observation process. 

It is well-known that
the conditional distribution $\eta_t$ of the signal state $X_t$ given $\Fa_t$ is a $r_1$-dimensional Gaussian  distribution with a 
a mean and covariance matrix 
$$
\widehat{X}_t:=\EE(X_t~|~\Fa_t)\quad\mbox{\rm and}\quad
P_t:=\EE\left(\left(X_t-\EE(X_t~|~\Fa_t)\right)\left(X_t-\EE(X_t~|~\Fa_t)\right)^{\prime}\right)
$$ 
given by the Kalman-Bucy and the Riccati equations 
\begin{eqnarray}
d\widehat{X}_t&=&A~\widehat{X}_t~dt+P_{t}~C^{\prime}R^{-1}_{2}~\left(dY_t-C\widehat{X}_tdt\right)\label{nonlinear-KB-mean}\\
\partial_tP_t&=&\ricc(P_t).\label{nonlinear-KB-Riccati}
\end{eqnarray}
with the Riccati drift function from $\SS^+_{r_1}$ into $\SS_{r_1}$ 
(where $\SS^+_{r_1}$ (resp.~$\SS_{r_1}$) is the collection of symmetric and positive definite (resp.~semi-definite) $r_1\times r_1$ matrices
defined for any $Q\in \SS^+_{r_1}$ ) by
\begin{equation}\label{def-ricc}
\ricc(Q)=AQ+QA^{\prime}-QSQ+R\quad\mbox{\rm with}\quad R=R_1\quad\mbox{\rm and}\quad S:=C^{\prime}R^{-1}_{2}C 
\end{equation}

\subsection{Nonlinear Kalman-Bucy Diffusions}

We now consider three conditional nonlinear McKean-Vlasov type diffusion processes
\begin{eqnarray}\label{Kalman-Bucy-filter-nonlinear-ref}
d\overline{X}_t & = &A~\overline{X}_t~dt~+~R^{1/2}_{1}~d\overline{W}_t+\Pa_{\eta_t}C^{\prime}R^{-1}_{2}~\left[dY_t-\left(C\overline{X}_tdt+R^{1/2}_{2}~d\overline{V}_{t}\right)\right]
\label{Kalman-Bucy-filter-nonlinear-ref}\\
d\overline{X}_t & = &A~\overline{X}_t~dt~+~R^{1/2}_{1}~d\overline{W}_t+\Pa_{\eta_t}C^{\prime}R^{-1}_{2}~\left[dY_t-\left(\frac{1}{2}C\left[\overline{X}_t+\eta_t(e)\right]dt\right)\right]
\label{denkf}\\
d\overline{X}_t & = &A~\overline{X}_t~dt~+~R_1\Pa_{\eta_t}^{-1}\left(\overline{X}_t-\eta_t(e)\right)~dt+\Pa_{\eta_t}C^{\prime}R^{-1}_{2}~\left[dY_t-\left(\frac{1}{2}C\left[\overline{X}_t+\eta_t(e)\right]dt\right)\right]
\label{otenkf}
\end{eqnarray}
where  $(\overline{W}_t,\overline{V}_t,\overline{X}_0)$ are independent copies of $(W_t,V_t,X_0)$ (thus independent of
 the signal and the observation path). In the above displayed formula $\Pa_{\eta_t}$ stands for the covariance matrix
\begin{equation}\label{def-nl-cov}
\Pa_{\eta_t}=\eta_t\left[(e-\eta_t(e))(e-\eta_t(e))^{\prime}\right]
\quad\mbox{\rm with}\quad \eta_t:=\mbox{\rm Law}(\overline{X}_t~|~\Fa_t)\quad\mbox{\rm and}\quad
e(x):=x
\end{equation}
and we use the notation, $\eta_t(f)=\int_{\mathbb{R}^{r_1}}f(x)\eta_t(dx)$ for $f:\mathbb{R}^{r_1}\rightarrow\mathbb{R}^{r}$ that is $\eta_t-$integrable (in particular,  $r=r_1$ in \eqref{denkf} and \eqref{otenkf} and $r=r_1^2$ in  \eqref{def-nl-cov}). 
Any of these probabilistic models will be commonly referred to as  Kalman-Bucy (nonlinear) diffusion processes. In the following we will denote by ${\mathcal G}_t$ the augmented filtration generated by  $\overline{X}_0$ and the triplet of independent Brownian motions $(Y_t,\overline{W}_t,\overline{V}_t)$. The process \eqref{Kalman-Bucy-filter-nonlinear-ref} corresponds to the vanilla type ensemble
Kalman-Bucy filter that is typically used in the literature. The process \eqref{denkf} is associated the so-called deterministic ensemble Kalman-Bucy filter \cite{sakov} and \eqref{otenkf} is a deterministic transport-inspired equation \cite{reich}.

We have the following result that is considered, for instance, in \cite{dt-2016}:
\begin{lem}
\label{lemmaKBnonlinear}Let $\overline{X}_{t}$ be a process such that 
$\mathbb{E}[\left\vert \overline{X}_{0}\right\vert^{2}]<\infty$
and that it satisfies any one of \eqref{Kalman-Bucy-filter-nonlinear-ref}-\eqref{otenkf}. Then the conditional mean and the
conditional covariance matrix (given ${\mathcal{F}}_{t}$) of any of the nonlinear
Kalman-Bucy diffusions \eqref{Kalman-Bucy-filter-nonlinear-ref}-\eqref{otenkf} satisfy
equations (\ref{nonlinear-KB-mean}) and (\ref{nonlinear-KB-Riccati}),
respectively.
\end{lem}

This result enables us to approximate  the mean and covariance associated to the linear filtering problem in
\eqref{lin-Gaussian-diffusion-filtering} by simulating $N$ i.i.d.~samples from one of the processes \eqref{Kalman-Bucy-filter-nonlinear-ref}-\eqref{otenkf} exactly
and then use the sample mean and sample covariance to approximate the mean and covariance of the filter. Since this exact simulation is seldom possible,
we consider how one can generate $N-$particle systems whose sample mean and sample covariance can be used to replace the exact simulation.

\begin{rem}
\label{lemmaKBnonlinear}Let us observe that equations \eqref{Kalman-Bucy-filter-nonlinear-ref}-\eqref{otenkf} have indeed a unique solution in the class of processes $Z_{t}$ such that $
\mathbb{E}[\sup_{t\in \lbrack 0,T]}\left\vert Z_{t}\right\vert^{2}]<\infty $. To show existence, one considers the (linear version of the) equations \eqref{Kalman-Bucy-filter-nonlinear-ref}-\eqref{otenkf}
with the conditional
expectations $\eta _{t}(e)$ and $\eta _{t}\left[ (e-\eta _{t}(e))(e-\eta
_{t}(e))^{\prime }\right] $  replaced by  the solution of the equations (\ref{nonlinear-KB-mean}) and (\ref{nonlinear-KB-Riccati}), respectively, that is, 

\begin{eqnarray}\label{Kalman-Bucy-filter-nonlinear-ref-lin}
d\overline{X}_t & = &A~\overline{X}_t~dt~+~R^{1/2}_{1}~d\overline{W}_t+P_tC^{\prime}R^{-1}_{2}~\left[dY_t-\left(C\overline{X}_tdt+R^{1/2}_{2}~d\overline{V}_{t}\right)\right]
\label{Kalman-Bucy-filter-nonlinear-ref-lin}\\
d\overline{X}_t & = &A~\overline{X}_t~dt~+~R^{1/2}_{1}~d\overline{W}_t+P_tC^{\prime}R^{-1}_{2}~\left[dY_t-\left(\frac{1}{2}C\left[\overline{X}_t+\eta_t(e)\right]dt\right)\right]
\label{denkf-lin}\\
d\overline{X}_t & = &A~\overline{X}_t~dt~+~R_1P_t^{-1}\left(\overline{X}_t-\hat{X}_t\right)~dt+P_tC^{\prime}R^{-1}_{2}~\left[dY_t-\left(\frac{1}{2}C\left[\overline{X}_t+\hat{X}_t\right]dt\right)\right].
\label{otenkf-lin}
\end{eqnarray}
Equations \eqref{Kalman-Bucy-filter-nonlinear-ref-lin}-\eqref{otenkf-lin}
have a unique solution (as they are linear) that indeed satisfy the corresponding  (nonlinear) equations \eqref{Kalman-Bucy-filter-nonlinear-ref}-\eqref{otenkf}.
To show the uniqueness of the solutions of equations \eqref{Kalman-Bucy-filter-nonlinear-ref}-\eqref{otenkf}, one observes first  that any solution $\eta_t$ of the (nonlinear) equations \eqref{Kalman-Bucy-filter-nonlinear-ref}-\eqref{otenkf} has its conditional
expectations $\eta _{t}(e)$ and $\eta _{t}\left[ (e-\eta _{t}(e))(e-\eta
_{t}(e))^{\prime }\right] $  uniquely characterized by  the equations (\ref{nonlinear-KB-mean}) and (\ref{nonlinear-KB-Riccati}). Therefore they satisfy the corresponding linear versions of the equations \eqref{Kalman-Bucy-filter-nonlinear-ref}-\eqref{otenkf},
with the conditional
expectations $\eta _{t}(e)$ and $\eta _{t}\left[ (e-\eta _{t}(e))(e-\eta
_{t}(e))^{\prime }\right] $  replaced by  the solution of the equations (\ref{nonlinear-KB-mean}) and (\ref{nonlinear-KB-Riccati}). In other words, they satisfy equations \eqref{Kalman-Bucy-filter-nonlinear-ref-lin}-\eqref{otenkf-lin} and therefore they are  unique as the corresponding linear equations  \eqref{Kalman-Bucy-filter-nonlinear-ref-lin}-\eqref{otenkf-lin} have a unique solution.       
\end{rem}

\begin{rem}\label{rem:non_lin}
If one modifies \eqref{lin-Gaussian-diffusion-filtering} to
\begin{equation*}
\left\{
\begin{array}{rcl}
dX_t&=&f(X_t)~dt~+~R^{1/2}_{1}~dW_t\\
dY_t&=&C~X_t~dt~+~R^{1/2}_{2}~dV_{t}.
\end{array}
\right.
\end{equation*}
for some non-linear function $f:\mathbb{R}^{r_1}\rightarrow\mathbb{R}^{r_1}$, one can consider
a modification of any of \eqref{Kalman-Bucy-filter-nonlinear-ref}-\eqref{otenkf}. For instance in the case
\eqref{Kalman-Bucy-filter-nonlinear-ref}
$$
d\overline{X}_t = f(\overline{X}_t)~dt~+~R^{1/2}_{1}~d\overline{W}_t+\Pa_{\eta_t}C^{\prime}R^{-1}_{2}~\left[dY_t-\left(C\overline{X}_tdt+R^{1/2}_{2}~d\overline{V}_{t}\right)\right].
$$
We note however that any approximation of this process, as alluded to above, does not typically provide an approximation of the non-linear filter. Nonetheless it is considered
in many works in the field.
\end{rem}

\subsection{Ensemble Kalman-Bucy Filters}\label{ref-sec-theo}

We now consider three Ensemble Kalman-Bucy filters that correspond to the mean-field particle
interpretation of the nonlinear diffusion processes
\eqref{Kalman-Bucy-filter-nonlinear-ref}-\eqref{otenkf}. To be more precise we let $(%
\overline{W}_{t}^{i},\overline{V}_{t}^{i},\xi _{0}^{i})_{1\leq i\leq N}$ be $%
N$ independent copies of $(\overline{W}_{t},\overline{V}_{t},\overline{X}%
_{0})$. In this notation, we have the three McKean-Vlasov type
interacting diffusion processes for $i\in\{1,\dots,N\}$
\begin{eqnarray}
d\xi _{t}^{i} & = & A~\xi _{t}^{i}dt+R_{1}^{1/2}d\overline{W}%
_{t}^{i}+p_{t}C^{\prime }R_{2}^{-1}\left[ dY_{t}-\left( C\xi
_{t}^{i}dt+R_{2}^{1/2}~d\overline{V}_{t}^{i}\right) \right]  \label{fv1-3}\\
d\xi _{t}^{i} & = & A~\xi _{t}^{i}dt+R_{1}^{1/2}d\overline{W}%
_{t}^{i}+p_{t}C^{\prime }R_{2}^{-1}\left[ dY_{t}
-\left(\frac{1}{2}C\left[\xi_t^i+\eta_t^N(e)\right]dt\right)\right]  \label{denkf_p}\\
d\xi _{t}^{i} & = & A~\xi _{t}^{i}dt+ R_{1}(p_{t})^{-1}\left(\xi_t^i-\eta_t^N(e)\right)~dt 
+p_{t}C^{\prime }R_{2}^{-1}\left[ dY_{t}
-\left(\frac{1}{2}C\left[\xi_t^i+\eta_t^N(e)\right]dt\right)\right]  \label{otenkf_p}
\end{eqnarray}
with the rescaled particle covariance matrices 
\begin{equation}
p_{t}:=\left( 1-\frac{1}{N}\right) ^{-1}~\mathcal{P}_{\eta _{t}^{N}}=%
\frac{1}{N-1}\sum_{1\leq i\leq N}\left( \xi _{t}^{i}-m_{t}\right) \left( \xi
_{t}^{i}-m_{t}\right) ^{\prime }  \label{fv1-3-2}
\end{equation}%
and the empirical measures 
\begin{equation*}
\eta _{t}^{N}:=\frac{1}{N}\sum_{1\leq i\leq N}\delta _{\xi _{t}^{i}}\quad %
\mbox{\rm and the sample mean}\quad m_{t}:=\frac{1}{N}\sum_{1\leq i\leq
N}\xi _{t}^{i}.
\end{equation*}
Note that for \eqref{otenkf_p} one needs $N\geq r_1$ for the almost sure invertibility of $p_{t}$, but it is also sufficient to use the pseudo inverse instead of
the inverse.

These processes have been thoroughly studied in the literature and a review can be found in \cite{BD20}.
We have the evolution equations for \eqref{fv1-3}  and \eqref{denkf_p}
\[
\begin{array}{rcl}
dm_t&=&\displaystyle A~m_tdt+p_t~C^{\prime}\Sigma^{-1}~\left(dY_t-Cm_t~dt\right)+\frac{1}{\sqrt{N+1}}~d\overline{M}_t\\
&&\\
dp_t&=&\displaystyle\ricc(p_t)~dt+\frac{1}{\sqrt{N}}~dM_t
\end{array}
\]
with a triplet of pairwise orthogonal martingales $M_t$ and $\overline{M}_t$ and the innovation martingale $$\widehat{M}_t:=Y_t-C\int_{0}^t\widehat{X}_sds.
$$ 
Note that for any $t\geq 0$, the absolute moments of $\overline{M}_t$ and $M_t$ are uniformly, w.r.t.~$N$, upper-bounded.
We remark that the mathematical theory for \eqref{otenkf_p}  in the linear-Gaussian setting reverts to that of the standard Kalman-Bucy filter as detailed in
\cite{BD17}.

\begin{rem}
Returning to the context of Remark \ref{rem:non_lin}, one could, for instance, run the following version of \eqref{fv1-3} for $i\in\{1,\dots,N\}$
$$
d\xi _{t}^{i}  =  f(\xi _{t}^{i})dt+R_{1}^{1/2}d\overline{W}%
_{t}^{i}+p_{t}^{N}C^{\prime }R_{2}^{-1}\left[ dY_{t}-\left( C\xi
_{t}^{i}dt+R_{2}^{1/2}~d\overline{V}_{t}^{i}\right) \right]  
$$
again, noting that this will typically not produce a consistent approximation of the non-linear filter, as would be the case for \eqref{fv1-3} with the model
\eqref{lin-Gaussian-diffusion-filtering}.
\end{rem}

Below we denote $\|\cdot\|$ as the $L_2-$norm for vectors. For a square matrix $B$ say, $B_{\textrm{sym}}=\tfrac{1}{2}(B+B')$ and $\mu(B)$ being the largest eigenvalue of $B_{\textrm{sym}}$.
In the cases of \eqref{fv1-3} and \eqref{denkf_p}, \cite{dt-2016} consider the following assumption (recall \eqref{def-ricc}): we have $\mu(S)>0$ with
\begin{equation}\label{eq:ass1}
S = \mu(S) Id
\end{equation}
where $Id$ is the $r_1\times r_1$ identity matrix (in general we use $Id$ to denote the identity matrix of `appropriate dimension', appropriate depending on the context).
\cite{dt-2016} prove the following time-uniform convergence theorem for the mean.
\begin{theorem}\label{theo:enkbf}
Consider the cases of \eqref{fv1-3} and \eqref{denkf_p}. 
Assume that \eqref{eq:ass1} holds and that $\mu(A)<0$. Then for any $n\geq 1$ and $N$ sufficiently large, we have that
$$
\sup_{t\geq 0}\mathbb{E}[\|m_{t}-\widehat{X}_t\|^n]^{\tfrac{1}{n}} \leq \frac{\mathsf{C}(n)}{\sqrt{N}}
$$
where $\mathsf{C}(n)<\infty$ is a constant that does not depend upon $N$.
\end{theorem}

We denote by $P_{\infty}$ as the solution to $\ricc(P_{\infty})=0$.
Set $\mu(A-P_{\infty}S)$ as the largest eigenvalue of $\tfrac{1}{2}(A-P_{\infty}S+(A-P_{\infty}S)')$.
For the case of \eqref{denkf_p}, the following result, presented in \cite{BD20}, can be shown.
\begin{theorem}\label{theo:denkbf}
Consider the case of \eqref{denkf_p}. 
Assume that $\mu(A-P_{\infty}S)<0$ and $S\in\mathbb{S}_{r_1}^+$. Then for any $n\geq 1$, $N\geq 2$ there exists a $t(n,N)>0$ with $t(n,N)\rightarrow\infty$ as $N\rightarrow\infty$, such that for any $t\in[0,t(n,N)]$
$$
\mathbb{E}[\|m_{t}-\widehat{X}_t\|^n]^{\tfrac{1}{n}} \leq \frac{\mathsf{C}}{\sqrt{N}}
$$
where $\mathsf{C}<\infty$ is a constant that does not depend upon $N$. $t(n,N)$ is characterized in \cite{BD19}.
\end{theorem}

We note that $t(n,N)$ is $\mathcal{O}(\log(N))$; see \cite[Theorem 2.1]{BD19} when $\epsilon=1/\sqrt{N}$ ($\epsilon$ is as \cite[Theorem 2.1]{BD19}).

For the case of \eqref{otenkf_p} we have the following result. Below we cite  \cite[eq.~(2.7)]{BD20} and it should be noted that the notation of that article differs slightly to this one,
so we mean of course that the equation holds in the notation of this article. Also note that we use the inverse of  $p_{t}^{N}$, but as noted above this need not be the case and then
the constraint on $N$ below is not required. 
\begin{theorem}\label{theo:otenkbf}
Consider the case of \eqref{otenkf_p}.
Assume that \cite[eq.~(2.7)]{BD20} holds. Then for any $n\geq 1$ there exists a $\mathsf{C}<\infty$, $\kappa\in(0,1)$ such that for any $N\geq r_1$ and $t\geq 0$ we have that
$$
\mathbb{E}[\|m_{t}-\widehat{X}_t\|^n]^{\tfrac{1}{n}} \leq \frac{\mathsf{C}\kappa^t}{\sqrt{N}}.
$$
\end{theorem}

\subsection{Conditional Bias Result}

We set
$$
\widehat{m}_t:=\EE(m_t~|~\Fa_t)\qquad
 \widehat{p}_t:=\EE(p_t~|~\Fa_t).
$$
We consider bounds on
$$
\EE\left(\Vert \widehat{m}_t-\widehat{X}_t\Vert^n\right)^{1/n}.
$$
These bounds do not currently exist in the literature and will allow us to investigate our new estimator for the normalization constant, later on in the article.
In the case \eqref{otenkf_p} we have (see e.g.~\cite{BD17})
\begin{equation}\label{eq:otenkf_mm}
\widehat{m}_t=\widehat{X}_t\quad \mbox{\rm and}\quad \widehat{p}_t=P_t
\end{equation}
so we focus only on \eqref{fv1-3} and \eqref{denkf_p}. 

Using the second order Taylor-type expansions presented in \cite{BD-EJP20}  when the pair $(A,R_1^{1/2})$ is stabilizable and $(A,S^{1/2})$ is detectable we have the uniform almost sure bias and for any $n\geq 1$, the $\LL_n$-mean error estimates
\begin{equation}\label{u-bias-p}
0\leq P_t-\widehat{p}_t\leq \mathsf{C}_1/N\quad \mbox{\rm and}\quad 
\EE\left(\Vert P_t-p_t\Vert^n~|~\Fa_t\right)^{1/n}\leq \mathsf{C}_2(n)/\sqrt{N}
\end{equation}
as soon as $N$ is chosen sufficiently large, for some deterministic constants $\mathsf{C}_1,\mathsf{C}_2(n)$ that do not depend on the time horizon or $N$ (but $\mathsf{C}_2(n)$ does depend upon $n$). Whenever $\mu(A)<0$, Theorem \ref{theo:enkbf} yields the rather crude estimate
\begin{equation}\label{u-var-w-p}
\EE\left(\Vert \widehat{m}_t-\widehat{X}_t\Vert^n\right)^{1/n}=
\EE\left(\Vert \EE\left(m_t-\widehat{X}_t~|~\Fa_t\right)\Vert^n\right)^{1/n}\leq 
\EE\left(\Vert m_t-\widehat{X}_t\Vert^n\right)^{1/n} \leq  \frac{\mathsf{C}(n)}{\sqrt{N}}.
\end{equation}
Our objective is to improve the estimate (\ref{u-var-w-p}). We observe that
\begin{equation}\label{eqmt}
\begin{array}{l}
d(m_t-\widehat{X}_t)\\
\\
=\displaystyle A~(m_t-\widehat{X}_t)dt+(p_t-P_t)~C^{\prime}\Sigma^{-1}~\left(dY_t-Cm_t~dt\right)\\
\\
\displaystyle\hskip3cm+P_t~C^{\prime}\Sigma^{-1}~\left(dY_t-Cm_t~dt\right)-P_t~C^{\prime}\Sigma^{-1}~\left(dY_t-C\widehat{X}_t~dt\right)+\frac{1}{\sqrt{N+1}}~d\overline{M}_t
\end{array}
\end{equation}
which yields
\begin{equation}\label{eqmthat1}
\begin{array}{l}
d(\widehat{m}_t-\widehat{X}_t)\\
\\
=\displaystyle (A-P_tS)~(\widehat{m}_t-\widehat{X}_t)dt+(\widehat{p}_t-P_t)~C^{\prime}\Sigma^{-1}~\left(dY_t-C\widehat{X}_t~dt\right)+\EE\left((P_t-p_t)~S(m_t-\widehat{X}_t)~|~\Fa_t\right)~dt.
\end{array}
\end{equation}
Equation \eqref{eqmthat1} is deduced by conditioning both terms in \eqref{eqmt} with respect to $\Fa_t$, exchanging the (stochastic) integration with the conditional expectation for all the terms on the right hand side of \eqref{eqmt}. To justify this Fubini-like exchange one can use, for example, Lemma 3.21 in \cite{BC09}. One also needs to use the fact that, for arbitrary $0\le s \le t$, $\widehat{m}_s=\EE(m_s~|~\Fa_s)=\EE(m_s~|~\Fa_t)$ and $ \widehat{p}_s:=\EE(p_s~|~\Fa_s)=\EE(p_s~|~\Fa_t)$.
To justify this, one can use, for example, Proposition 3.15 in \cite{BC09}. We note that both Proposition 3.15 and Lemma 3.21 hold true when the conditioning is done with respect the equivalent probability measure $\tilde{\mathbb{P}}$ under which $Y$ is a Brownian motion. However since 
$(\overline{W}_{t},\overline{V}_{t},\overline{X}_{0})$ are independent of $Y$ under $\mathbb{P}$ they remain independent of $Y$ under $\tilde{\mathbb{P}}$ and therefore conditioning the processes $\xi^i$ under $\tilde{\mathbb{P}}$ is the same as 
conditioning them under $\mathbb{P}$. Note that argument this does \emph{not} apply to the original signal $X$.

Let $\Ea_{s,t}$ be the exponential semigroup (or the state transition matrix) associated with the smooth flow of matrices $t\mapsto (A-P_tS)$ defined for any $s\leq t$ by the forward and backward differential equations,
\begin{equation*}
 \partial_t \,\Ea_{s,t}=(A-P_tS)\,\Ea_{s,t}\quad\mbox{\rm and}\quad
\partial_s\, \Ea_{s,t}=-\Ea_{s,t}\,(A-P_sS)
\end{equation*}
with $\Ea_{s,s}=Id$. Equivalently in terms of the matrices
$\Ea_t:=\Ea_{0,t}$ we have
$
\Ea_{s,t}=\Ea_t\Ea_s^{-1}
$. Under some observability and controllability conditions, we have that the drift matrices $(A-P_tS)$
delivers a uniformly stable time varying linear system in the sense that
\begin{equation}\label{exp-decay-E}
\Vert \Ea_{s,t}\Vert\leq c~e^{-\lambda~(t-s)}\quad \mbox{\rm for some}\quad \lambda>0.
\end{equation}
See for instance Eq.~(2.13) in the review article~\cite{BD20}.
In this notation,
recalling that $\widehat{m}_0=\widehat{X}_0$ we have the bias formulae
$$
\widehat{m}_t-\widehat{X}_t=\int_0^t\Ea_{s,t}~
\underbrace{(\widehat{p}_s-P_s)}_{\in [-c_1/N,0]~a.e.}~C^{\prime}\Sigma^{-1}~\left(dY_s-C\widehat{X}_s~dt\right)+
\int_0^t\Ea_{s,t}~\EE\left((P_s-p_s)~S(m_s-\widehat{X}_s)~|~\Fa_s\right)~ds.
$$
Note that $(\widehat{p}_s-P_s)\in [-c_1/N,0]$~a.e.~can be justified via \eqref{u-bias-p}.
Combining the Burkholder-Davis-Gundy inequality with the almost sure and uniform bias estimate (\ref{u-bias-p}) and the exponential decay (\ref{exp-decay-E}) we obtain the uniform estimate
$$
\begin{array}{l}
\displaystyle\EE\left(\Vert\int_0^t\Ea_{s,t}~
(\widehat{p}_s-P_s)~C^{\prime}\Sigma^{-1}~\left(dY_s-C\widehat{X}_s~dt\right)\Vert^n\right)^{1/n}\displaystyle\leq \frac{\mathsf{C}(n)}{N}
\end{array}
$$
for some deterministic constants $\mathsf{C}(n)$ that do not depend upon $t$ nor $N$.
Conversely, combining H\"older's inequality with the uniform variance estimate (\ref{u-bias-p}) we have the almost sure inequality
$$
\EE\left(\Vert\EE\left((P_s-p_s)~S(m_s-\widehat{X}_s)~|~\Fa_s\right)\Vert^n\right)^{1/n}
\leq \frac{\mathsf{C}(n)}{N}~$$
for some deterministic constants $\mathsf{C}(n)$ that do not depend upon $t$ nor $N$.
Applying the generalized Minkowski inequality we have thus proved  the following theorem.

\begin{theorem}\label{theo:cond_bias}
Consider the cases of \eqref{fv1-3} and \eqref{denkf_p}. Assume that \eqref{eq:ass1} holds and that $\mu(A)<0$. Then for any $n\geq 1$ 
and $N$ sufficiently large, we have 
$$
\sup_{t\geq 0}\EE\left(\Vert \widehat{m}_t-\widehat{X}_t\Vert^n\right)^{1/n}\leq  \frac{\mathsf{C}(n)}{N}
$$
where $\mathsf{C}(n)<\infty$ is a constant that does not depend upon $N$.
\end{theorem}

\section{Computing the Normalizing Constant}\label{sec:nc}

\subsection{Estimation}

We define $Z_t:=\frac{{\mathcal L}_{X_{0:t},Y_{0:t}}}{{\mathcal L}_{X_{0:t},W_{0:t}}}$ to be the density of ${\mathcal L}_{X_{0:t},Y_{0:t}}$, the law of the process $(X,Y)$ and that of ${\mathcal L}_{X_{0:t},W_{0:t}}$, the law of the process $(X,W)$. That is, 
\[
{\mathbb E}[f(X_{0:t})g(Y_{0:t})]= {\mathbb E}[f(X_{0:t})g(W_{0:t})Z_{t}(X,Y)].
\]
One can show that (see Exercise 3.14 pp 56 in \cite{BC09})
$$
Z_t(X,Y)=\exp{\left[\int_0^t \left[\langle CX_s, R_2^{-1} dY_s\rangle -\frac{1}{2}\langle X_s,SX_s\rangle ~ds\right]\right]}.
$$
Following a standard approach (see, e.g., Chapter 3 in \cite{BC09}), we introduce a
probability measure $\tilde{\mathbb{P}}$ 
by specifying its Radon--Nikodym derivative with respect to $\mathbb{P}$
to be given by $Z_t(X,Y)^{-1}$, i.e.
\[
\left. \frac{{\mathrm d} \tilde{\mathbb{P}}}{{\mathrm d} \mathbb{P}}\right| _{\mathcal{G}_t}=Z_t.
\]
Under $\tilde{\mathbb{P}}$, the observation process $Y$ is a scaled Brownian motion
independent of $X$; additionally the law of the signal process $X$ under
$\tilde{\mathbb{P}}$ is the same as its law under $\mathbb{P}$\footnote{The proof at this statement is an immediate application of Girsanov's theorem and follows in a similar manner with the proof of Proposition 3.13 in \cite{BC09}.}.
Moreover, for every  $f$ defined on the signal path space, we have  
the Bayes' formula 
\begin{equation}
\EE\left(f((X_s)_{s\in [0,t]})~\vert~\Fa_t\right)=\frac{\tilde \EE\left(f((X_s)_{s\in [0,t]})~Z_t(X,Y)~\vert~\Fa_t\right)}{\tilde\EE\left(Z_t(X,Y)~\vert~\Fa_t\right)}.
\label{kallstrie}
\end{equation}
where $\tilde \EE()$ means expectation with respect to $\tilde{\mathbb{P}}$\footnote{
Formula (\ref{kallstrie}) is commonly known as the Kallianpur-Striebel formula. For a proof, see, e.g., Prop 3.16 in \cite{BC09}.}. Since, under $\tilde{\mathbb{P}}$, $X$ and $Y$ are independent we can interpret the numerator  $\tilde \EE\left(f((X_s)_{s\in [0,t]})~Z_t(X,Y)~\vert~\Fa_t\right)$ and 
$\tilde\EE\left(Z_t(X,Y)~\vert~\Fa_t\right)$ as integrals with respect to the law of the signal, where the integrant has the observation process path fixed $Y$.\footnote{This interpretation can be made rigorous, see e.g. the robust representation formulation of the filtering problem described in Chapter 5 in \cite{BC09}.}. We can therefore, let $\overline{Z}_t(Y)$ be the  likelihood function defined by
$$
\overline{Z}_t(Y):=\EE_Y\left(Z_t(X,Y)\right),
$$
where $\EE_Y\left(\point\right)$ stands for the expectation w.r.t. the signal process when the observation is fixed and independent of the signal.
In this notation, the Bayes' formula (\ref{kallstrie}) takes the form
$$
\EE\left(f((X_s)_{s\in [0,t]})~\vert~\Fa_t\right)=\frac{\EE_Y\left(f((X_s)_{s\in [0,t]})~Z_t(X,Y)\right)}{\EE_Y\left(Z_t(X,Y)\right)}.
$$

We have
\begin{eqnarray*}
dZ_t(X,Y)&=&\left[\langle CX_t, R_2^{-1} dY_t\rangle-\frac{1}{2}\langle X_t,SX_t\rangle~dt\right]Z_t+\frac{1}{2}~\langle X_t,SX_t\rangle~Z_t~dt\\
&=&Z_t(X,Y)~\langle CX_t,R_2^{-1}  dY_t\rangle
\end{eqnarray*}
We check this claim using the fact that
\begin{eqnarray*}
d\langle CX_t,R_2^{-1}dY_t\rangle d\langle CX_t,R_2^{-1} dY_t\rangle&=&\sum_{k,l} (CX_t)(k)(R^{-1/2}_2dV_t)(k)(CX_t)(l)(R^{-1/2}_2dV_t)(l)\\
&=&\sum_{k,l,k^{\prime},l^{\prime}} (CX_t)(k)~R^{-1/2}_1(k,k^{\prime})dV_t(k^{\prime})
(CX_t)(l)R_2^{-1/2}(l,l^{\prime})dV_t(l^{\prime})\\
&=&\sum_{k,l,k^{\prime}} (CX_t)(k)~R_2^{-1}(k,l)
(CX_t)(l)dt=\langle X_t,SX_t\rangle~dt
\end{eqnarray*}
This implies that
$$
\overline{Z}_t(Y)=1+\int_0^t~\overline{Z}_s(Y)~\overline{Z}_s(Y)^{-1}\EE_Y\left(Z_s(X,Y)~\langle CX_s, R_2^{-1} dY_s\rangle\right)
=1+\int_0^t~\overline{Z}_s(Y)~\langle C\widehat{X}_s,R_2^{-1}  dY_s\rangle
$$
from which we conclude that
$$
\overline{Z}_t(Y)=\exp{\left[\int_0^t 
\left[~\langle C\widehat{X}_s, R_2^{-1} dY_s\rangle -\frac{1}{2}\langle \widehat{X}_s,S\widehat{X}_s\rangle~ds\right]
\right]}
$$
or equivalently
$$
d\overline{Z}_t(Y)=\overline{Z}_t(Y)~\langle C\widehat{X}_t, R_2^{-1} dY_t\rangle.
$$
This suggests the estimator:
\begin{equation}\label{est}
\overline{Z}_t^N(Y)=\exp{\left[\int_0^t 
\left[~\langle Cm_s, R_2^{-1} dY_s\rangle -\frac{1}{2}\langle m_s,Sm_s\rangle~ds\right]
\right]}.
\end{equation}
which satisfies the equation 
\begin{equation}\label{eqest}
d\overline{Z}_t(Y)=\overline{Z}_t(Y)~\langle Cm_t, R_2^{-1} dY_t\rangle.
\end{equation}

We remark that any of the three processes \eqref{fv1-3}-\eqref{otenkf_p} can 
be used to compute $\overline{Z}_t^N(Y)$. 

\subsection{Justification of the Estimator}

We now seek to justify the estimator $\overline{Z}_t^N(Y)$.
We have
$$
\begin{array}{l}
\overline{Z}_t(X,Y)\\
\\
:=\overline{Z}_t(Y)^{-1}Z_t(X,Y)\\
\\
=\exp{\left[\int_0^t \left[\langle C(X_s-\widehat{X}_s), R_2^{-1} (dY_s-C\widehat{X}_sds)\rangle -
\left[\frac{1}{2}\langle X_s,SX_s\rangle-\frac{1}{2}\langle \widehat{X}_s,S\widehat{X}_s\rangle-\langle X_s-\widehat{X}_s,S\widehat{X}_s\rangle\right] ~ds \right]\right]}\\
\\
=\exp{\left[\int_0^t \left[\langle C(X_s-\widehat{X}_s), R_2^{-1} (dY_s-C\widehat{X}_sds)\rangle -\frac{1}{2}
\langle (X_s-\widehat{X}_s),S(X_s-\widehat{X}_s)\rangle ~ds\right]\right]}
\end{array}$$
from which we conclude that
$$
d\overline{Z}_t(X,Y)=\overline{Z}_t(X,Y)~\langle C(X_s-\widehat{X}_s), R_2^{-1} (dY_s-C\widehat{X}_sds)\rangle.
$$
This already implies that
$$
\EE\left[\overline{Z}_t(X,Y)~|~\Fa_t\right]=1\Longleftrightarrow\EE\left[Z_t(X,Y)~|~\Fa_t\right]=\overline{Z}_t(Y).
$$

Now, we have
$$
\begin{array}{l}
\overline{Z}_t^N(Y)\overline{Z}_t^{-1}(Y)\\
\\
=\exp{\left[\int_0^t 
\left[~\langle C(m_s-\widehat{X}_s), R_2^{-1} dY_s\rangle -\frac{1}{2}\left[\langle m_s,Sm_s\rangle-\langle \widehat{X}_s,S\widehat{X}_s\rangle\right]~ds\right]
\right]}\\
\\
=\exp{\left[\int_0^t 
\left[~\langle C(m_s-\widehat{X}_s), R_2^{-1} dY_s\rangle -\frac{1}{2}\langle (m_s-\widehat{X}_s),S(m_s-\widehat{X}_s)\rangle ~ds\right] \right]}~\\
\\
\times\exp{\left[\int_0^t 
\left[~\frac{1}{2}\langle (m_s-\widehat{X}_s),S(m_s-\widehat{X}_s)\rangle-
\frac{1}{2}\left[\langle m_s,Sm_s\rangle-\langle \widehat{X}_s,S\widehat{X}_s\rangle ~ds \right]\right]
\right]}\\
\\
=\exp{\left[\int_0^t 
\left[~\langle C(m_s-\widehat{X}_s), R_2^{-1} dY_s\rangle -\frac{1}{2}\langle (m_s-\widehat{X}_s),S(m_s-\widehat{X}_s)\rangle ~ds \right] \right]}
\times\exp{\left[\int_0^t 
\langle \widehat{X}_s-m_s,S\widehat{X}_s\rangle ds\right]}\\
\\
=\exp{\left[\int_0^t 
\left[~\langle C(m_s-\widehat{X}_s), R_2^{-1} dY_s\rangle -\frac{1}{2}\langle (m_s-\widehat{X}_s),S(m_s-\widehat{X}_s)\rangle ~ds\right] \right]}
\times\exp{\left[-\int_0^t 
\langle C(m_s-\widehat{X}_s),R^{-1}C\widehat{X}_s\rangle ds\right]}\\
\\
=\exp{\left[\int_0^t 
\left[~\langle C(m_s-\widehat{X}_s), R_2^{-1} (dY_s-C\widehat{X}_sds)\rangle -\frac{1}{2}\langle (m_s-\widehat{X}_s),S(m_s-\widehat{X}_s)\rangle ~ds \right] \right]}
\end{array}$$
Observe that $dY_s-C\widehat{X}_sds$ is an $\Fa_s$-martingale increment. 
We also have
\begin{equation}\label{ref-Z-U}
d(\overline{Z}_t^N(Y)\overline{Z}_t^{-1}(Y))=\overline{Z}_t^N(Y)\overline{Z}_t^{-1}(Y)~
\langle C(m_t-\widehat{X}_t), R_2^{-1} (dY_t-C\widehat{X}_tdt)\rangle. 
\end{equation}

Now to conclude, it follows that that $\overline{Z}_t^N(Y)\overline{Z}_t^{-1}(Y)$ is a positive local martingale and therefore a supermartingale, hence  
\[
\EE (\overline{Z}_t^N(Y)\overline{Z}_t^{-1}(Y))\le \EE (\overline{Z}_0^N(Y)\overline{Z}_0^{-1}(Y))=1
\] 
for all $t\ge 0$. To justify the martingale property of $\overline{Z}_t^N(Y)\overline{Z}_t^{-1}(Y)$ one can use, for example, an argument based on Corollary 5.14 in \cite{ks}. 
As a result of the above calculations we can deduce that $\overline{Z}_t^N(Y)$ is in some sense well-defined, but a biased estimator of the normalization constant. Therefore,
we will focus on the logarithm of the normalization constant, as it is this latter quantity that is used in practical algorithms. We begin by investigating the conditional bias.


\subsection{Conditional Bias}

We now consider the estimate of the logarithm of the normalizing constant
$$
U_t^N(Y) := \log(\overline{Z}_t^N(Y))
$$
with the notation $U_t(Y)=\log(\overline{Z}_t^N(Y))$. 

Set $\widehat{U}_t(Y):=\EE(U_t^N(Y)~|~\Fa_t)$. Then, using (\ref{ref-Z-U}) we have
$$
\begin{array}{l}
U_t^N(Y)-U_t(Y)\\
\\
= \log(\overline{Z}_t^N(Y)/\overline{Z}_t(Y))\\
\\
\displaystyle=\int_0^t 
\left[~\langle C(m_s-\widehat{X}_s), R_2^{-1} (dY_s-C\widehat{X}_sds)\rangle -\frac{1}{2}\langle (m_s-\widehat{X}_s),S(m_s-\widehat{X}_s)\rangle ~ds\right].
\end{array}
$$
This yields the bias formula
\begin{equation}\label{u-bias}
\begin{array}{l}
\widehat{U}_t(Y)-U_t(Y)\\
\\
\displaystyle=\int_0^t 
\left[~\langle C(\widehat{m}_s-\widehat{X}_s), R_2^{-1} (dY_s-C\widehat{X}_sds)\rangle -\frac{1}{2}~\EE\left(\langle (m_s-\widehat{X}_s),S(m_s-\widehat{X}_s)\rangle ~|~\Fa_s\right)~ds \right]. 
\end{array}
\end{equation}
Taking the expectation in the above display, when (\ref{eq:ass1}) holds and $\mu(A)<0$ and applying Theorem \ref{theo:enkbf} we readily check that
$$
0\leq \EE\left(U_t(Y)\right)-\EE\left(\widehat{U}_t(Y)\right)=\frac{\mu(S)}{2}~\int_0^t \EE\left(\Vert m_s-\widehat{X}_s\Vert^2\right)~ds \leq \frac{\mathsf{C}t}{N}
$$
which is the bias, but is not of particular interest. So we will focus on the conditional bias $\EE\left(\left[\widehat{U}_t(Y)-U_t(Y)\right]^n\right)^{1/n}$.

In the case \eqref{otenkf_p}, recall \eqref{eq:otenkf_mm}.
In this context, using  the generalized Minkowski inequality and under the assumptions of Theorem~\ref{theo:otenkbf} we find that
$$
\sup_{t\geq 0}{\EE\left(\left|\widehat{U}_t(Y)-U_t(Y)\right]^n\right|^{1/n}}\leq \frac{\mathsf{C}(n)}{N}
$$
where $\mathsf{C}(n)<\infty$ is a constant that does not depend upon $N$.

For the cases of \eqref{fv1-3} and \eqref{denkf_p}, we start with the fact that the conditional bias in (\ref{u-bias}) is decomposed into two terms
$$
\alpha_t:=\int_0^t 
\langle C(\widehat{m}_t-\widehat{X}_s), R_2^{-1} (dY_s-C\widehat{X}_sds)\rangle \quad
\mbox{\rm and}
\quad
\beta_t:= 
\frac{1}{2}~\int_0^t \EE\left(\langle (m_s-\widehat{X}_s),S(m_s-\widehat{X}_s)\rangle ~|~\Fa_s\right)~ds.
$$
Using the uniform estimates presented in Section~\ref{ref-sec-theo} we have
$$
\EE\left(\left|\beta_t\right|^n\right)^{1/n}\leq \frac{\mathsf{C}(n)t}{N}
$$
for some deterministic constants $\mathsf{C}(n)$ that do not depend on the time horizon. On the other hand, combining the Burkholder-Davis-Gundy inequality with the uniform bias estimate (\ref{u-var-w-p}) we have the rather crude estimate
\begin{equation}\label{ref-crude-est}
\EE\left(\left|\alpha_t\right|^n\right)\leq \mathsf{C}(n)\left(\frac{t}{N}\right)^{n/2}
\end{equation}
for some deterministic constants $\mathsf{C}(n)$ that do not depend on the time horizon.
Arguing as for the proof of Theorem \ref{theo:cond_bias}, we check that
$$
\EE\left(\left|\alpha_t\right|^n\right)\leq \mathsf{C}(n)\left(\frac{\sqrt{t}}{N}\right)^n.
$$
Observe that the above improves the crude estimate stated in (\ref{ref-crude-est}).
This yields the following corollary.

\begin{cor}
Consider the cases of \eqref{fv1-3} and \eqref{denkf_p}. Assume that \eqref{eq:ass1} holds and that $\mu(A)<0$. Then for any $n\geq 1$, $t\geq 0$ and $N$ sufficiently large, we have that
$$
\EE\left(\left|\widehat{U}_t(Y)-U_t(Y)\right|^n\right)^{1/n}\leq \frac{\mathsf{C}(n)(t+\sqrt{t})}{N}
$$
where $\mathsf{C}(n)<\infty$ is a constant that does not depend upon $t$ or $N$.
\end{cor}

\subsection{$\mathbb{L}_n-$Error}

We now consider the $\mathbb{L}_n-$error of the estimate of the log-of the normalizing constant
$
U_t^N(Y) = \log(\overline{Z}_t^N(Y)).
$

In the case of \eqref{fv1-3} and \eqref{denkf_p} we have the following.
\begin{prop}\label{prop:log_var_enkbf}
Consider the cases of \eqref{fv1-3} and \eqref{denkf_p}. 
Assume that \eqref{eq:ass1} holds and that $\mu(A)<0$. Then for any $n\geq 1$, $t\geq 0$ and $N$ sufficiently large, we have that
$$
\mathbb{E}\left(\left|U_t^N(Y)-U_t(Y)\right|^n\right)^{1/n} \leq \mathsf{C}(n)\Big(\sqrt{\frac{t}{N}} + \frac{t}{N}\Big)
$$
where $\mathsf{C}(n)<\infty$ is a constant that does not depend upon $t$ or $N$.
\end{prop}

\begin{proof}
We have that
$$
U_t^N(Y)-U_t(Y) = \int_0^t \Big[\langle C(m_s - \widehat{X}_s), R_2^{-1} (dY_s-C\widehat{X}_sds)\rangle
-\frac{1}{2}\langle (m_s-\widehat{X}_s),S(m_s-\widehat{X}_s)\rangle~ds\Big].
$$ 
So one can apply the Minkowski inequality to yield that
$$
\mathbb{E}\left(\left|U_t^N(Y)-U_t(Y)\right|^n\right)^{1/n} \leq 
$$
\begin{equation}\label{eq:mast_eq}
\mathbb{E}\left(\Big|\int_0^t \langle C(m_s - \widehat{X}_s), R_2^{-1} (dY_s-C\widehat{X}_sds)\rangle ds\Big|^n\right)^{1/n} + 
\frac{1}{2}~\mathbb{E}\left(\Big|\int_0^t\langle (m_s-\widehat{X}_s),S(m_s-\widehat{X}_s)\rangle ds\Big|^n\right)^{1/n}.
\end{equation}
For the left term on the R.H.S.~of \eqref{eq:mast_eq} one can use the Burkholder-Gundy-Davis inequality along with Theorem \ref{theo:enkbf}
to yield that
$$
\mathbb{E}\left(\Big|\int_0^t \langle C(m_s - \widehat{X}_s), R_2^{-1} (dY_s-C\widehat{X}_sds)\rangle ds\Big|^n\right)^{1/n} \leq
\frac{\mathsf{C}(n)t^{\frac{1}{2}}}{\sqrt{N}}
$$
for some $\mathsf{C}<\infty$ a constant that does not depend upon $t$ or $N$.
For the right term on the R.H.S.~of \eqref{eq:mast_eq} one can use the generalized Minkowski inequality and Theorem \ref{theo:enkbf}
to give
$$
\mathbb{E}\left(\Big|\int_0^t\langle (m_s-\widehat{X}_s),S(m_s-\widehat{X}_s)\rangle ds\Big|^n\right)^{1/n}\leq
\frac{\mathsf{C}(n)t}{N}
$$
for some $\mathsf{C}(n)<\infty$ a constant that does not depend upon $t$ or $N$. The proof is now easily concluded.
\end{proof}

\noindent In the case of \eqref{denkf_p} we can refine further:
\begin{prop}\label{prop:log_var_denkbf}
Consider the case of \eqref{denkf_p}. 
Assume that $\mu(A-P_{\infty}S)<0$ and $S\in\mathbb{S}_{r_1}^+$. Then for any $N\geq 2$ there exists a $t(N)>0$ with $t(N)\rightarrow\infty$ as $N\rightarrow\infty$, such that for any $t\in[0,t(N)]$
$$
\mathbb{E}\left(\left|U_t^N(Y)-U_t(Y)\right|^n\right)^{1/n} \leq \mathsf{C}(n)\Big(\sqrt{\frac{t}{N}} + \frac{t}{N}\Big)
$$
where $\mathsf{C}(n)<\infty$ is a constant that does not depend upon $t$ or $N$.
\end{prop}
\begin{proof}
The proof is essentially as Proposition \ref{prop:log_var_enkbf} except one should use Theorem \ref{theo:denkbf} instead of Theorem \ref{theo:enkbf}.
\end{proof}

\noindent In the case of \eqref{otenkf_p} we have the following, whose proof is again similar to the above results.
\begin{prop}\label{prop:log_var_otenkbf}
Consider the case of \eqref{otenkf_p} . 
Assume that \cite[eq.~(2.7)]{BD20} holds. Then for any $t\geq 0$ and $N\geq r_1$, we have that
$$
\mathbb{E}\left(\left|U_t^N(Y)-U_t(Y)\right|^n\right)^{1/n} \leq \frac{\mathsf{C}(n)}{\sqrt{N}}
$$
where $\mathsf{C}(n)<\infty$ is a constant that does not depend upon $t$ or $N$.
\end{prop}

Both Proposition \ref{prop:log_var_enkbf} and Proposition \ref{prop:log_var_denkbf} establish that one can estimate the log of the normalization constant 
using the ensemble Kalman-Bucy type filters, with a mean square error (for instance) that grows at most linearly in time. Proposition \ref{prop:log_var_otenkbf} provides a uniform in time
error, mostly following from the deterministic nature of the algorithm and the dependence on standard Kalman-Bucy theory.

One can say more, when considering the average estimation error in 
over a window of time $w$ as we now state. We restrict ourselves to the mean square error below.
In the cases of both \eqref{fv1-3} and \eqref{denkf_p} we have the time-uniform upper-bound.
\begin{cor}\label{cor:enkbf}
Consider the case of \eqref{fv1-3}  and \eqref{denkf_p}. 
Assume that \eqref{eq:ass1} holds and that $\mu(A)<0$. Then for any $t>0$, $0<w<t$ and $N$ sufficiently large, we have that
$$
\mathbb{E}\Bigg[\Bigg(\frac{1}{w}\Bigg\{\log\Bigg(\frac{\overline{Z}_t^N(Y)}{\overline{Z}_{t-w}^N(Y)}\Bigg)-\log\Bigg(\frac{\overline{Z}_t(Y)}{\overline{Z}_{t-w}(Y)}\Bigg)\Bigg\}\Bigg)^2\Bigg] \leq \frac{\mathsf{C}}{N}
$$
where $\mathsf{C}<\infty$ is a constant that does not depend upon $t$, $w$ or $N$.
\end{cor}

In the case of \eqref{denkf_p} we refine further and have the following time-uniform upper-bound.
\begin{cor}\label{cor:denkbf}
Consider the cases of \eqref{denkf_p}. 
Assume that $\mu(A-P_{\infty}S)<0$ and $S\in\mathbb{S}_{r_1}^+$. Then for any $N\geq 2$ there exists a $t(N)>0$ with $t(N)\rightarrow\infty$ as $N\rightarrow\infty$, such that for any $t\in(0,t(N)]$, $0<w<t$
$$
\mathbb{E}\Bigg[\Bigg(\frac{1}{w}\Bigg\{\log\Bigg(\frac{\overline{Z}_t^N(Y)}{\overline{Z}_{t-w}^N(Y)}\Bigg)-\log\Bigg(\frac{\overline{Z}_t(Y)}{\overline{Z}_{t-w}(Y)}\Bigg)\Bigg\}\Bigg)^2\Bigg] \leq \frac{\mathsf{C}}{N}
$$
where $\mathsf{C}<\infty$ is a constant that does not depend upon $t$, $w$ or $N$.
\end{cor}

Note that, in Corollary \ref{cor:denkbf} as we require $t\in(0,t(N)]$ and $t(N)=\mathcal{O}(\log(N))$ this may not be as useful as Corollary \ref{cor:enkbf}, but the assumption of a stable system in Corollary \ref{cor:enkbf} is much stronger than the hypotheses in Corollary \ref{cor:denkbf}; see \cite{BD20} for a discussion.

%
%

\subsection{Simulation Results}

For $(i,k,L)\in \{1,\cdots,N\} \times \mathbb{N}_0\times  \mathbb{N}_0$, let $\Delta_L=2^{-L}$ and consider the Euler-discretization of \eqref{fv1-3}-\eqref{otenkf_p}:
\begin{equation}
\renewcommand{\arraystretch}{1.7}
\begin{array}{lccl}
(\textbf{F1})  & \xi_{(k+1)\Delta_L}^i & = & \xi_{k\Delta_L}^i + A  \xi_{k\Delta_L}^i  \Delta_L + R_1^{1/2} \big\{\overline{W}_{(k+1)\Delta_L}^i-\overline{W}_{k\Delta_L}^i \big\} + \\
& & & p_{k\Delta_L} C' R_2^{-1} \Big( \big\{Y_{(k+1)\Delta_L} - Y_{k\Delta_L} \big\} - \Big[ C \xi_{k\Delta_L}^i  \Delta_L + R_2^{1/2} \big \{\overline{V}_{(k+1)\Delta_L}^i - \overline{V}_{k\Delta_L}^i \big\} \Big] \Big)\\
(\textbf{F2})  & \xi_{(k+1)\Delta_L}^i & = & \xi_{k\Delta_L}^i + A  \xi_{k\Delta_L}^i  \Delta_L + R_1^{1/2} \big\{\overline{W}_{(k+1)\Delta_L}^i-\overline{W}_{k\Delta_L}^i \big\} + \\
& & & p_{k\Delta_L} C' R_2^{-1} \left( \big\{Y_{(k+1)\Delta_L} - Y_{k\Delta_L} \big\} - C \left(\dfrac{ \xi_{k\Delta_L}^i + m_{k\Delta_L}}{2} \right)\Delta_L \right)\\
(\textbf{F3})  & \xi_{(k+1)\Delta_L}^i & = & \xi_{k\Delta_L}^i + A  \xi_{k\Delta_L}^i  \Delta_L + R_1 \left(p_{k\Delta_L}\right)^{-1} \left(\xi_{k\Delta_L}^i - m_{k\Delta_L}  \right) \Delta_L +\\
& & & p_{k\Delta_L} C' R_2^{-1} \left( \big\{Y_{(k+1)\Delta_L} - Y_{k\Delta_L} \big\} - C \left(\dfrac{ \xi_{k\Delta_L}^i + m_{k\Delta_L}}{2} \right)\Delta_L \right)
\end{array}
\label{disc_EnKBF}
\end{equation}
and the discretization of $\overline{Z}_T^N(Y)$:
 \begin{align}
\overline{Z}_T^{N,L}(Y)  = \exp\left\{ \sum_{k=0}^{T/\Delta_L-1} \left\langle m_{k\Delta_L}, C' R_2^{-1}  \big[ Y_{(k+1)\Delta_L} - Y_{k\Delta_L} \big]\right\rangle - \frac{1}{2} \left\langle m_{k\Delta_L},  S\, m_{k\Delta_L} \right\rangle \Delta_L \right\}.
\label{eq:disc_log_likelihood}
\end{align}
For the purpose of showing that the mean square error of the estimate of the log of the normalization constant in the cases \textbf{F1} \& \textbf{F2} is of $\mathcal{O}(\frac{t}{N})$ and in the case \textbf{F3} is of $\mathcal{O}(\frac{1}{N})$, we take $r_1=r_2=1$, $A=-2$, $R_1^{-1/2}=1$, $R_2^{-1/2}=2$, $C$ a uniform random number in $(0,1]$ and $\xi_0^i \overset{i.i.d.}{\sim}  \mathcal{N}(0.5,0.2)$ (normal distribution mean 0.5 and variance 0.2). In Tables \ref{tab:var_F1} - \ref{tab:var_F2}  and Figures \ref{fig:var_F1} -  \ref{fig:var_F2}, we show that the rate of the MSE of the estimate in \eqref{eq:disc_log_likelihood} for the cases \textbf{F1} \& \textbf{F2} is of $\mathcal{O}(\frac{t}{N})$, even though we have used a naive time discretization for our results. In \autoref{tab:var_F3} and \autoref{fig:var_F3} we show that the rate in the \textbf{F3} case is of $\mathcal{O}(1/N)$.

\begin{table}[H]
\centering
\begin{tabular}{c|c|c|c|c|c|c|}
\cline{2-7}
                           & \multicolumn{2}{c|}{$N=1000$} & \multicolumn{2}{c|}{$N= 500$}          & \multicolumn{2}{c|}{$N=250$}           \\ \hline
\multicolumn{1}{|c|}{t}    & MSE       & MSE$/(t/N)$    & \multicolumn{1}{c|}{MSE} & MSE$/(t/N)$ & \multicolumn{1}{c|}{MSE} & MSE$/(t/N)$ \\ \hline
\multicolumn{1}{|c|}{50}   & 2.194E-04    & 4.4E-03        & 5.404E-04                & 5.4E-03     & 8.903E-04                & 4.5E-03     \\ \hline
\multicolumn{1}{|c|}{100}  & 5.483E-04    & 5.5E-03        & 1.442E-03                & 7.2E-03     & 3.082E-03                & 7.7E-03     \\ \hline
\multicolumn{1}{|c|}{200}  & 1.085E-03    & 5.4E-03        & 2.901E-03                & 7.3E-03     & 4.442E-03                & 5.6E-03     \\ \hline
\multicolumn{1}{|c|}{400}  & 2.283E-03    & 5.7E-03        & 5.047E-03                & 6.3E-03     & 1.017E-02                & 6.4E-03     \\ \hline
\multicolumn{1}{|c|}{800}  & 4.894E-03    & 6.1E-03        & 8.579E-03                & 5.4E-03     & 1.985E-02                & 6.2E-03     \\ \hline
\multicolumn{1}{|c|}{1600} & 9.716E-03    & 6.1E-03        & 2.039E-02                & 6.4E-03     & 2.904E-02                & 4.5E-03     \\ \hline
\multicolumn{1}{|c|}{3200} & 1.974E-02    & 6.2E-03        & 3.504E-02                & 5.5E-03     & 7.835E-02                & 6.1E-03     \\ \hline
\multicolumn{1}{|c|}{6400} & 3.571E-02    & 5.6E-03        & 7.087E-02                & 5.5E-03     & 1.599E-01                & 6.2E-03     \\ \hline
\end{tabular}
\caption{\label{tab:var_F1}The mean square error (MSE) and MSE$/(t/N)$ for $N\in\{250, 500, 1000\}$ in the \textbf{F1} case.}
\end{table}

\begin{table}[H]
\centering
\begin{tabular}{c|c|c|c|c|c|c|}
\cline{2-7}
                           & \multicolumn{2}{c|}{$N=1000$}          & \multicolumn{2}{c|}{$N= 500$}          & \multicolumn{2}{c|}{$N=250$} \\ \hline
\multicolumn{1}{|c|}{t}    & \multicolumn{1}{c|}{MSE} & MSE$/(t/N)$ & \multicolumn{1}{c|}{MSE} & MSE$/(t/N)$ & MSE          & MSE$/(t/N)$   \\ \hline
\multicolumn{1}{|c|}{50}   & 1.137E-03                & 5.7E-03     & 5.127E-04                & 5.1E-03     & 3.317E-04    & 6.6E-03       \\ \hline
\multicolumn{1}{|c|}{100}  & 2.157E-03                & 5.4E-03     & 1.128E-03                & 5.6E-03     & 8.892E-04    & 8.9E-03       \\ \hline
\multicolumn{1}{|c|}{200}  & 4.601E-03                & 5.8E-03     & 2.708E-03                & 6.8E-03     & 1.260E-03    & 6.3E-03       \\ \hline
\multicolumn{1}{|c|}{400}  & 9.777E-03                & 6.1E-03     & 5.250E-03                & 6.6E-03     & 2.436E-03    & 6.1E-03       \\ \hline
\multicolumn{1}{|c|}{800}  & 1.979E-02                & 6.2E-03     & 9.949E-03                & 6.2E-03     & 4.639E-03    & 5.8E-03       \\ \hline
\multicolumn{1}{|c|}{1600} & 4.901E-02                & 7.7E-03     & 1.595E-02                & 5.0E-03     & 8.895E-03    & 5.6E-03       \\ \hline
\multicolumn{1}{|c|}{3200} & 9.593E-02                & 7.5E-03     & 3.001E-02                & 4.7E-03     & 1.925E-02    & 6.0E-03       \\ \hline
\multicolumn{1}{|c|}{6400} & 1.744E-01                & 6.8E-03     & 5.415E-02                & 4.2E-03     & 3.635E-02    & 5.7E-03       \\ \hline
\end{tabular}
\caption{\label{tab:var_F2}The MSE and MSE$/(t/N)$ for $N\in\{250, 500, 1000\}$ in the \textbf{F2} case.}
\end{table}

\begin{table}[H]
\centering
\begin{tabular}{|c|c|c|}
\hline
\multicolumn{1}{|l|} {$N$} & MSE     & \multicolumn{1}{l|}{MSE$\times N$} \\ \hline
50                      & 1.073E-05 &   5.4E-04                   \\ \hline
100                     & 4.951E-06 & 5.0E-04                  \\ \hline
200                     & 2.867E-06 &5.7E-04                  \\ \hline
400                     & 1.492E-06 & 6.0E-04                   \\ \hline
800                     & 8.157E-07 & 6.5E-04                  \\ \hline
1600                    & 3.119E-07 & 5.0E-04                   \\ \hline
3200                    & 1.773E-07 & 5.7E-04                   \\ \hline
6400                    & 6.344E-08      & 4.1E-04                   \\ \hline
\end{tabular}
\caption{\label{tab:var_F3}The MSE and MSE$\times N$ for $t=100$ in the \textbf{F3} case.}
\end{table}

\begin{figure} [H]
\centering
\includegraphics[height=0.35\textwidth]{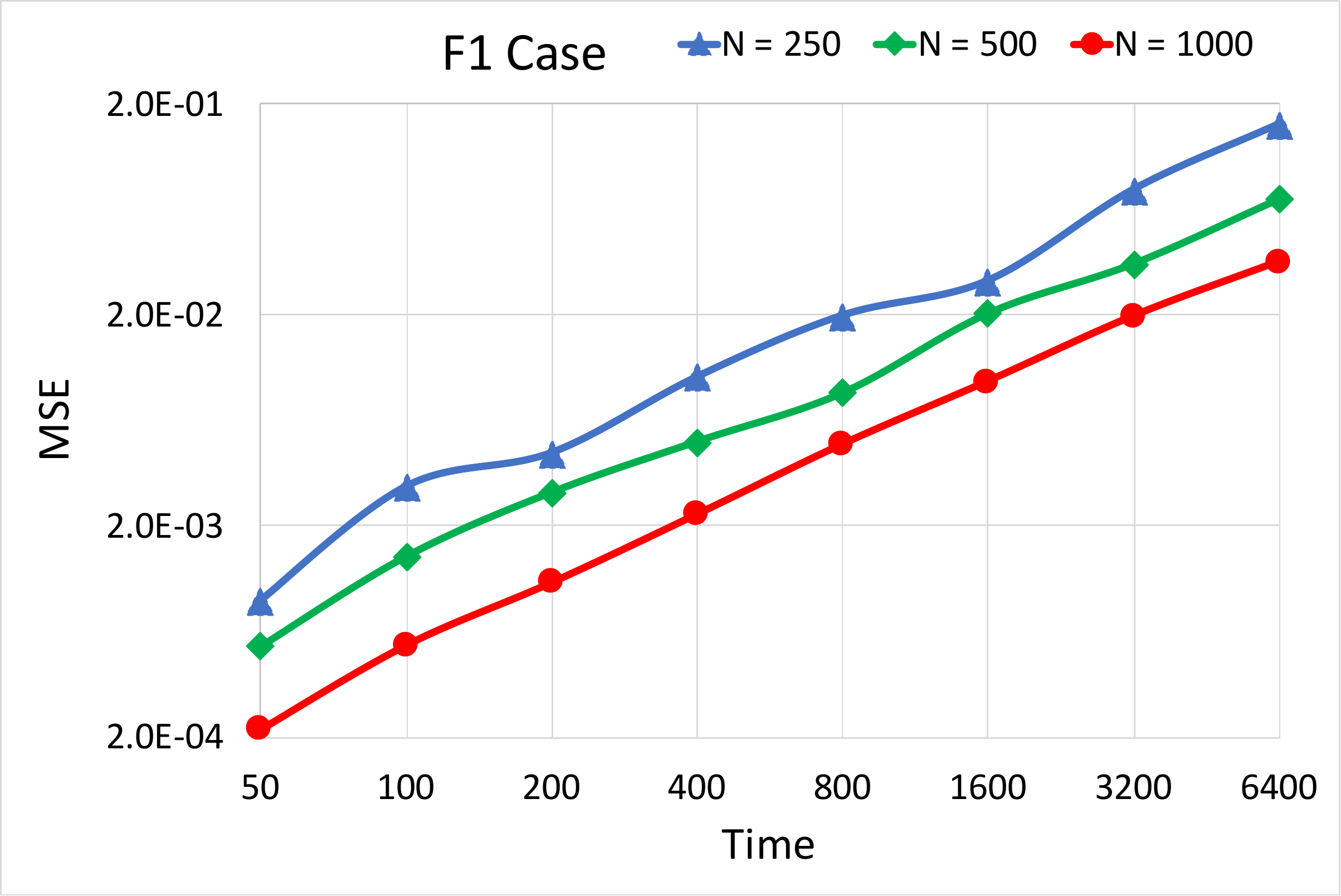} 
   \caption{The mean square error associated to the EnKBF in the \textbf{F1} case. This plots the MSE against the time parameter for $N\in\{250, 500, 1000\}$. MSE $\approx$ 5.9E-03 $\left(\frac{t}{N}\right)$, $\approx$ 6.1E-03 $\left(\frac{t}{N}\right)$, $\approx$ 5.6E-03 $\left(\frac{t}{N}\right)$, for $N=250, 500, 1000$, respectively.}
    \label{fig:var_F1}
\end{figure}

\begin{figure} [H]
\centering
\includegraphics[height=0.35\textwidth]{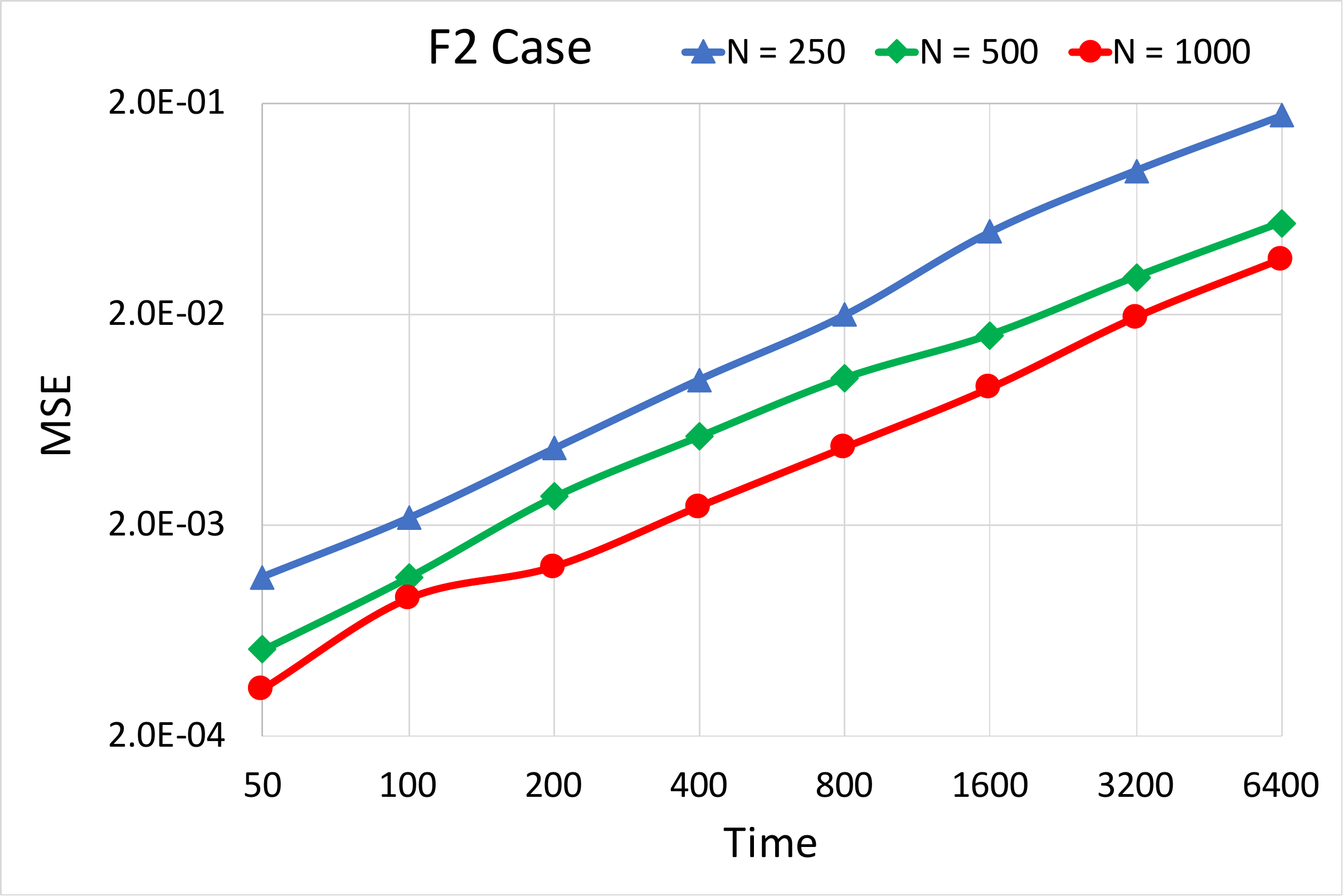} 
   \caption{The mean square error associated to the EnKBF in the \textbf{F2} case. This plots the MSE against the time parameter for $N\in\{250, 500, 1000\}$. MSE $\approx$ 6.4E-03 $\left(\frac{t}{N}\right)$, $\approx$ 5.5E-03 $\left(\frac{t}{N}\right)$, $\approx$ 6.4E-03 $\left(\frac{t}{N}\right)$, for $N=250, 500, 1000$, respectively.}
    \label{fig:var_F2}
\end{figure}
\begin{figure} [H]
\centering
\includegraphics[height=0.35\textwidth]{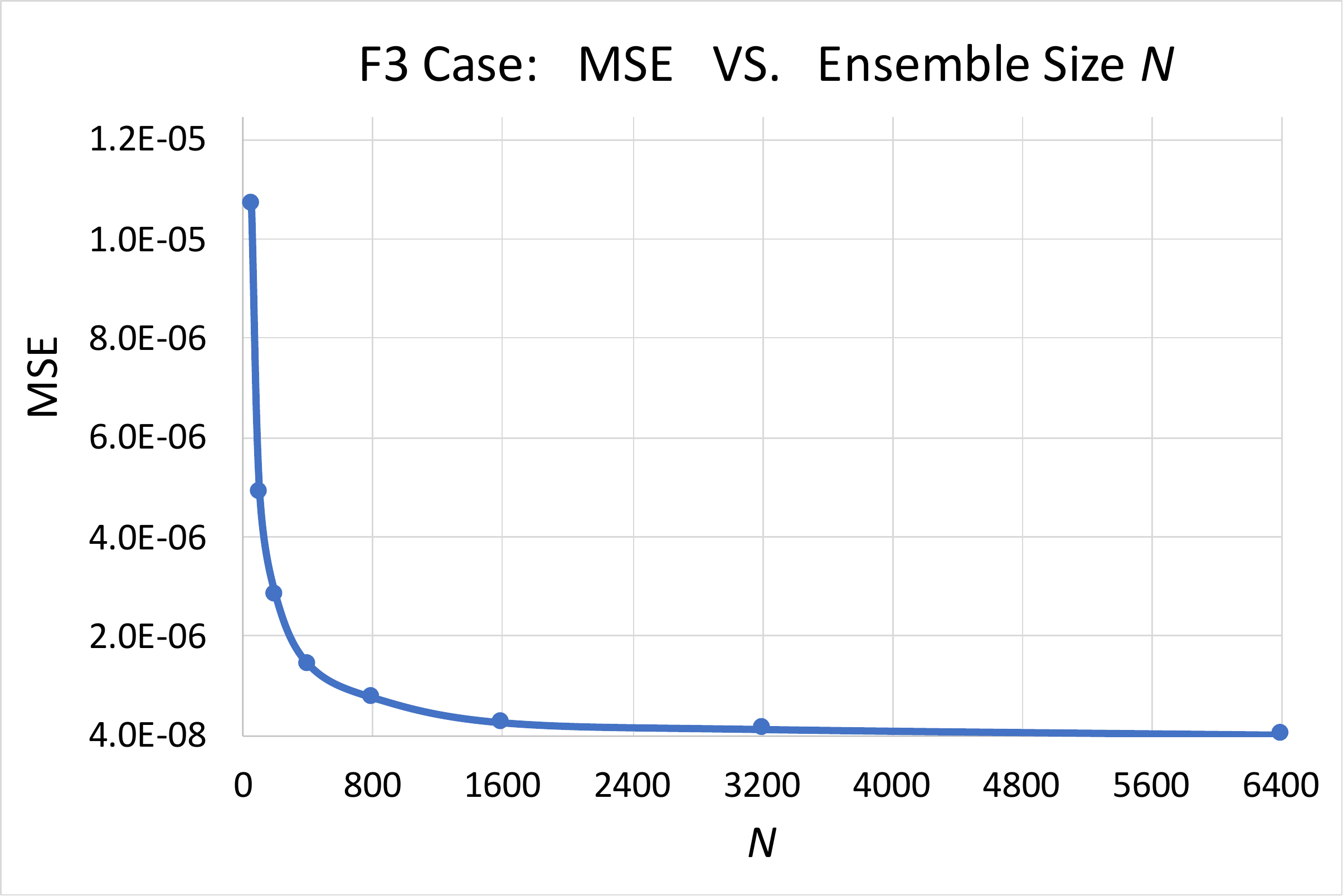} 
   \caption{The mean square error associated to the EnKBF in the \textbf{F3} case. This plots the MSE against the ensemble size $N$ for fixed time $t=100$. MSE $\approx $ 5.4E-04$/N$.}
    \label{fig:var_F3}
\end{figure}

\section{Application to Static Parameter Estimation}\label{sec:static_par}

\subsection{Approach}\label{sec:app_par}

We now assume that there are a collection of unknown parameters, $\theta\in\Theta\subseteq\mathbb{R}^{d_{\theta}}$, associated to the model \eqref{lin-Gaussian-diffusion-filtering}.
For instance $\theta$ could consist of some or all of the elements of $A, C, R_1, R_2$. We will then include an additional subscript $\theta$ in each of the mathematical objects that have been introduced in the previous two sections. As an example, we would write
$$
Z_{t,\theta}(X,Y)=\exp{\left[\int_0^t \left[\langle CX_s, R_2^{-1} dY_s\rangle -\frac{1}{2}\langle X_s,SX_s\rangle\right]~ds\right]}.
$$
For $0<s<t$ we introduce the notation
$$
\overline{Z}_{s,t,\theta}(Y) := \frac{\overline{Z}_{t,\theta}(Y)}{\overline{Z}_{s,\theta}(Y)}
$$
with the obvious extension to $\overline{Z}_{s,t,\theta}^N(Y):=\overline{Z}_{t,\theta}^N(Y)/\overline{Z}_{s,\theta}^N(Y)$.

The basic idea of our approach is to consider the recursive estimation of $\theta$ on the basis of the arriving observations. In particular, for notational convenience, we shall
consider a method which will update our current estimate of $\theta$ at unit times. Our objective is to follow a recursive maximum likelihood (RML) method (e.g.~\cite{legland1997})  which is based upon the
following update at any time $t\in\mathbb{N}$
$$
\theta_t = \theta_{t-1} + \kappa_t \nabla_{\theta}\log\Big(\overline{Z}_{t-1,t,\theta}(Y)\Big)\Big|_{\theta=\theta_{t-1}}
$$
where $\{\kappa_t\}_{t\in\mathbb{N}}$ is a sequence of real numbers with $\kappa_t>0$, $\sum_{t\in\mathbb{N}}\kappa_t=\infty$, $\sum_{t\in\mathbb{N}}\kappa_t^2<\infty$.
Computing the gradient of $\overline{Z}_{t-1,t,\theta}(Y)$ can be computationally expensive, so our approach is to use a type of finite difference estimator via SPSA.
We note that a classical finite difference estimator would require $2d_{\theta}$ evaluations of $\overline{Z}_{t-1,t,\theta}(Y)$, whereas the SPSA method keeps this evaluation
down to 2; as computational cost is a concern, we prefer this afore-mentioned method.

Our approach is the following, noting that we will use the argument $(k)$ to denote the $k^{th}-$element of a vector.
\begin{enumerate}
\item{Initialization: Set an initial $\theta_0\in\Theta$ and choose two step sizes $\{\kappa_t\}_{t\in\mathbb{N}}$ and $\{\nu_t\}_{t\in\mathbb{N}}$
such that $\kappa_t>0$, $\kappa_t,\nu_t\rightarrow 0$, $\sum_{t\in\mathbb{N}}\kappa_t=\infty$, $\sum_{t\in\mathbb{N}}\frac{\kappa_t^2}{\nu_t^2}<\infty$. Set $t=1$
and generate i.i.d.~the initial ensemble from $\eta_{0,\theta_0}$.
}
\item{Iteration: 
\begin{itemize}
\item{For $k\in\{1,\dots,d_{\theta}\}$, independently, sample $\Delta_t(k)$ from a Bernoulli distribution with success probability $1/2$ and support $\{-1,1\}$.}
\item{Set $\theta_{t-1}^+ = \theta_{t-1}+\nu_t\Delta_t$, $\theta_{t-1}^+ = \theta_{t-1}-\nu_t\Delta_t$.} 
\item{Using the EnKBF, with samples simulated under $\theta_{t-1}$, generate estimates $\overline{Z}_{t-1,t,\theta_{t-1}^+}^N(Y)$ and $\overline{Z}_{t-1,t,\theta_{t-1}^-}^N(Y)$.}
\item{Set, for $k\in\{1,\dots,d_{\theta}\}$,
$$
\theta_t(k) = \theta_{t-1}(k) + \kappa_t \frac{1}{2\nu_t\Delta_t(k)}\log\Bigg(\frac{\overline{Z}_{t-1,t,\theta_{t-1}^+}^N(Y)}{\overline{Z}_{t-1,t,\theta_{t-1}^-}^N(Y)}\Bigg).
$$
}
\item{Run the EnKBF from $t-1$ (starting with samples simulated under $\theta_{t-1}$) up-to time $t$ using the parameter $\theta_t$. Set $t=t+1$ and return to 2..}
\end{itemize}
}
\end{enumerate}
We note that in practice, one must run a time-discretization of the EnKBF, rather than the EnKBF itself. We remark also that the use of SPSA for parameter estimation associated to
hidden Markov models has been considered previously, for instance in \cite{poya}.

\subsection{Simulation Results}

We use the algorithm described in the previous section along with the Euler-discretization in \eqref{disc_EnKBF} to estimate the parameters in three different models. In particular, we show that the algorithm works in both linear and non-linear models. In all three models, the data is generated from the true parameters.

\subsubsection{Linear Gaussian Model}
In the first example, we take $A=\theta_1 Id$, $R_1^{-1/2}=\theta_2 R$, where 
\begin{align*}
R = \begin{bmatrix} 1 & 0.5 \\0.5& 1 & 0.5 \\ & 0.5 & \ddots & \ddots \\ & & \ddots & \ddots & 0.5 \\ & & & 0.5 & 1 \end{bmatrix} 
\end{align*}
$C = \alpha_1(r_1,r_2) C^*$, where $C^*$ is a uniform random matrix and $\alpha_1(r_1,r_2)$ is a constant, $R_2^{-1/2}=\alpha_2(r_2) Id$, where $\alpha_2(r_2)$ is a constant. 
In Figures \ref{fig:F1_r1=r2=2} - \ref{fig:F3_r1=r2=100}, we show the results for the parameters estimation of $\theta_1$ and $\theta_2$ in the cases $r_1=r_2=2$ and $r_1=r_2=100$. In all cases we take $N=100$ except in the case when $r_1=r_2=100$ in \textbf{F3}, where we take $N=200$ to assure the invertibility of $p_t^N$, otherwise, the condition number of $p_t^N$ will be huge. The discretization level is $L=8$ in all cases. The initial state is $X_0\sim \mathcal{N}(4\mathbf{1}_{r_1},Id)$, where $\mathbf{1}_{r_1}$ is a vector of 1's in $\mathbb{R}^{r_1}$.

The results display that in a reasonable case, one can estimate low-dimensional parameters using RML via SPSA. We now consider a few nonlinear models.

\begin{figure} [H]
\centering
\subfloat{ \includegraphics[clip, trim=0.7cm 0.2cm 0.4cm 0.1cm, width=8cm, height=5cm]{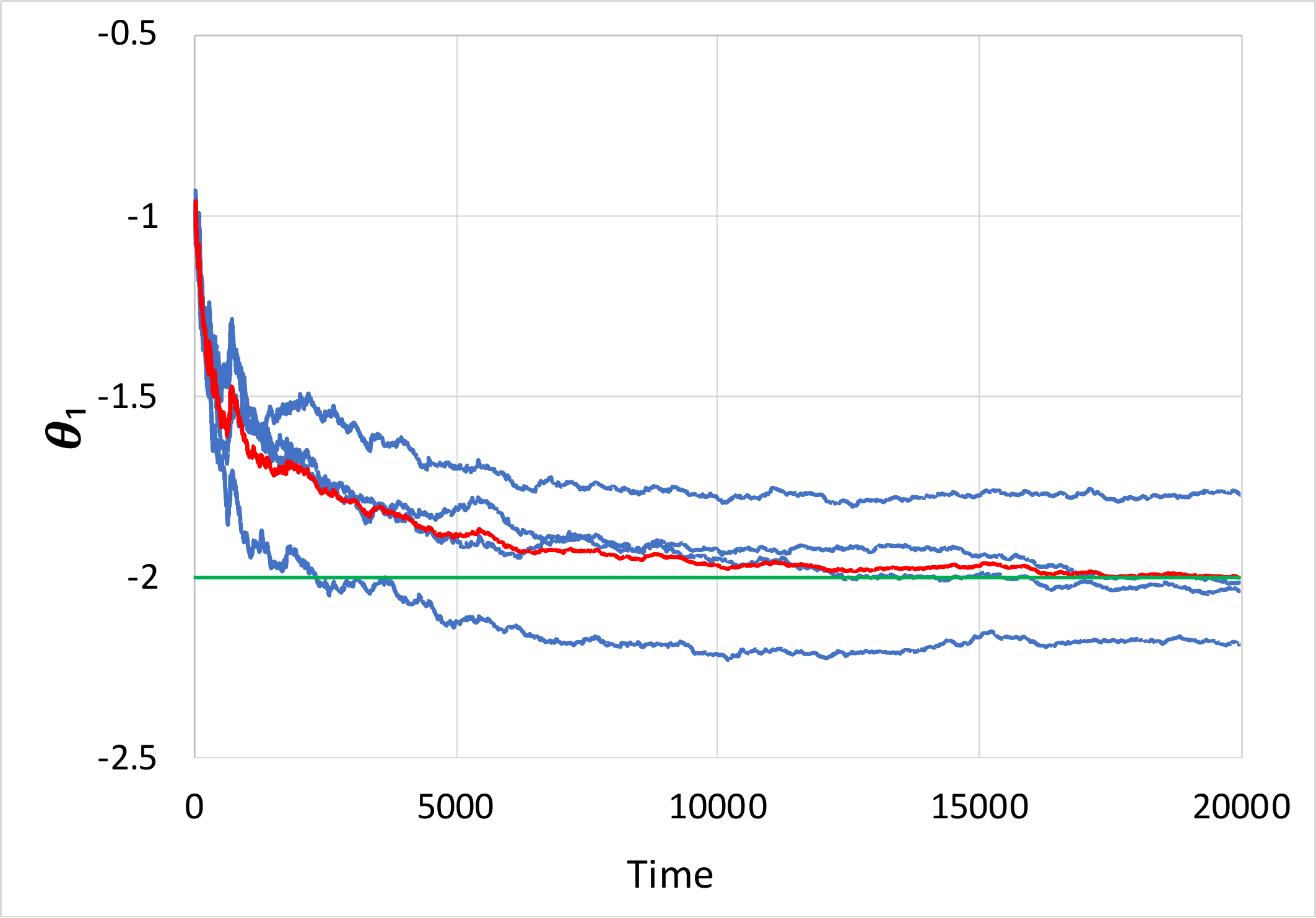} }
\subfloat{ \includegraphics[clip, trim=0.7cm 0.2cm 0.4cm 0.1cm, width=8cm, height=5cm]{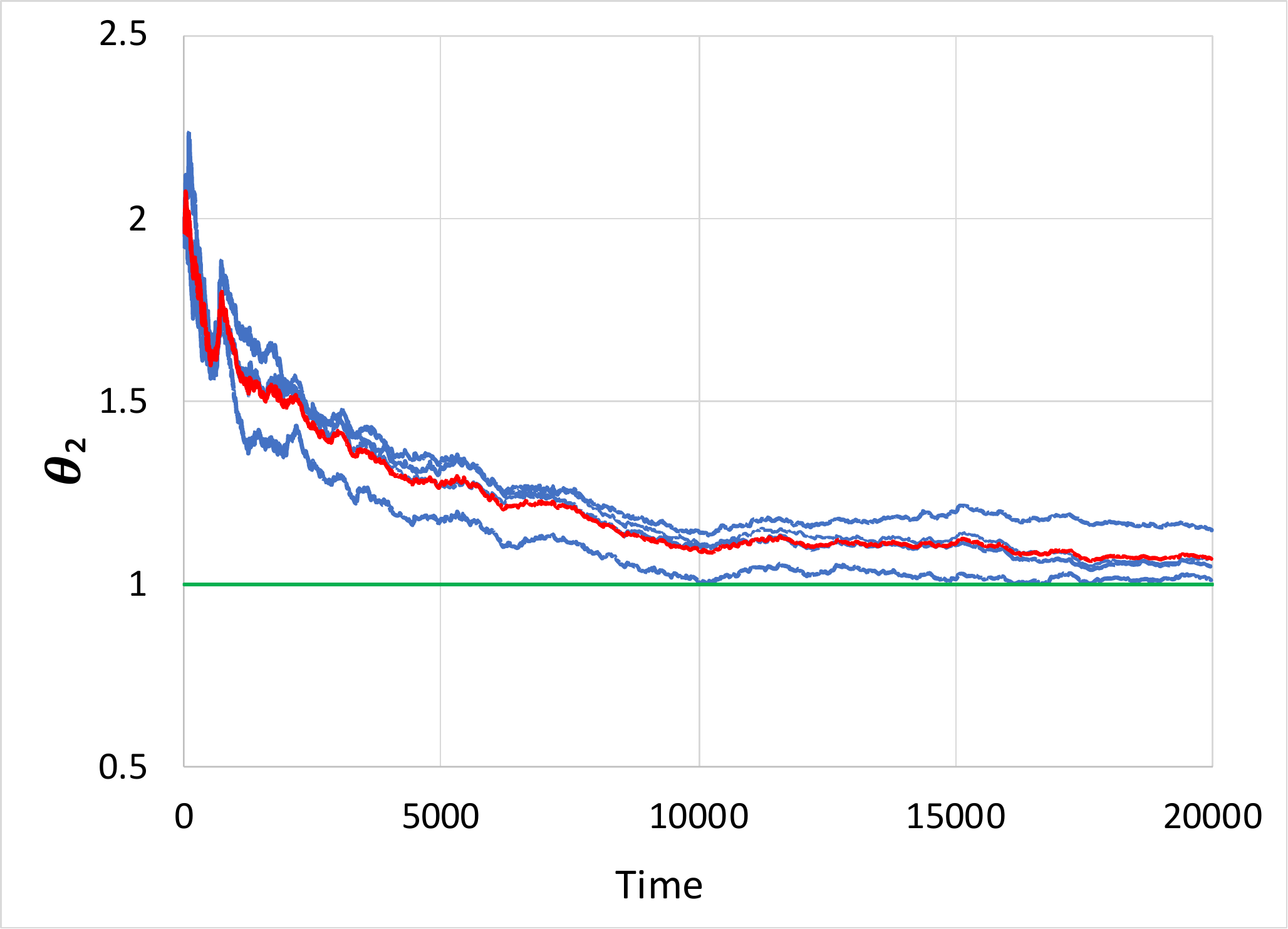} }
   \caption{The blue curves along with their average (in red) are trajectories from the execution of the algorithm in Section \ref{sec:app_par} for the estimation of $(\theta_1,\theta_2)$ in the case \textbf{F1} with $r_1=r_2=2$. The initial values of the parameters are $(-1,2)$. The green horizontal lines represent the true parameter values $(\theta_1^*,\theta_2^*)=(-2,1)$. We take $\nu_t=t^{-0.1}$ and $\kappa_t = 0.09$ when $t\leq 400$ and $\kappa_t=3 \times t^{-0.601}$ otherwise. }
    \label{fig:F1_r1=r2=2}
\end{figure}

\begin{figure} [H]
\centering
\subfloat{ \includegraphics[clip, trim=0.3cm 0.2cm 0.2cm 0.1cm,width=8cm, height=5cm]{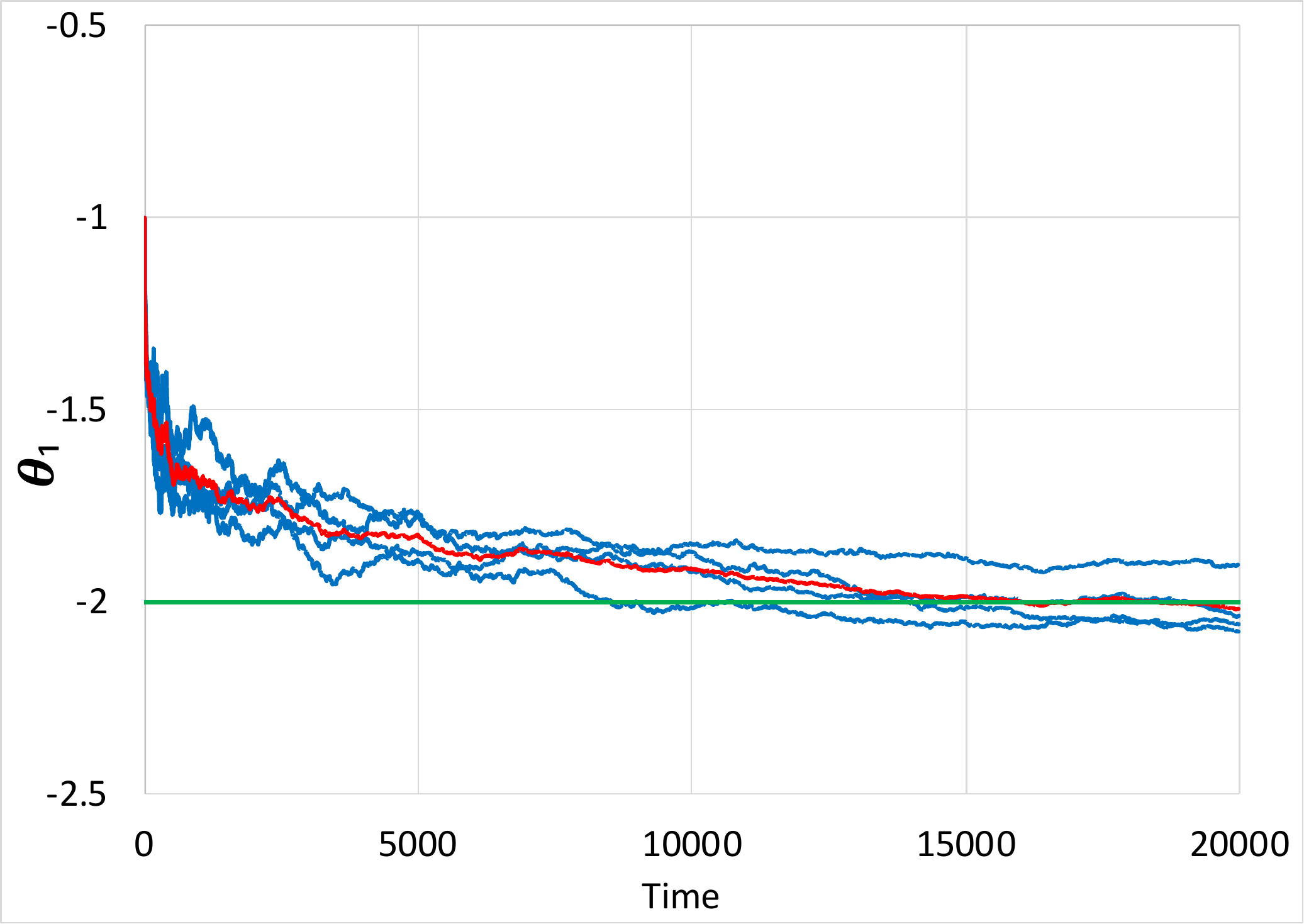} }
\subfloat{ \includegraphics[clip, trim=0.4cm 0.2cm 0.2cm 0.1cm,width=8cm, height=5cm]{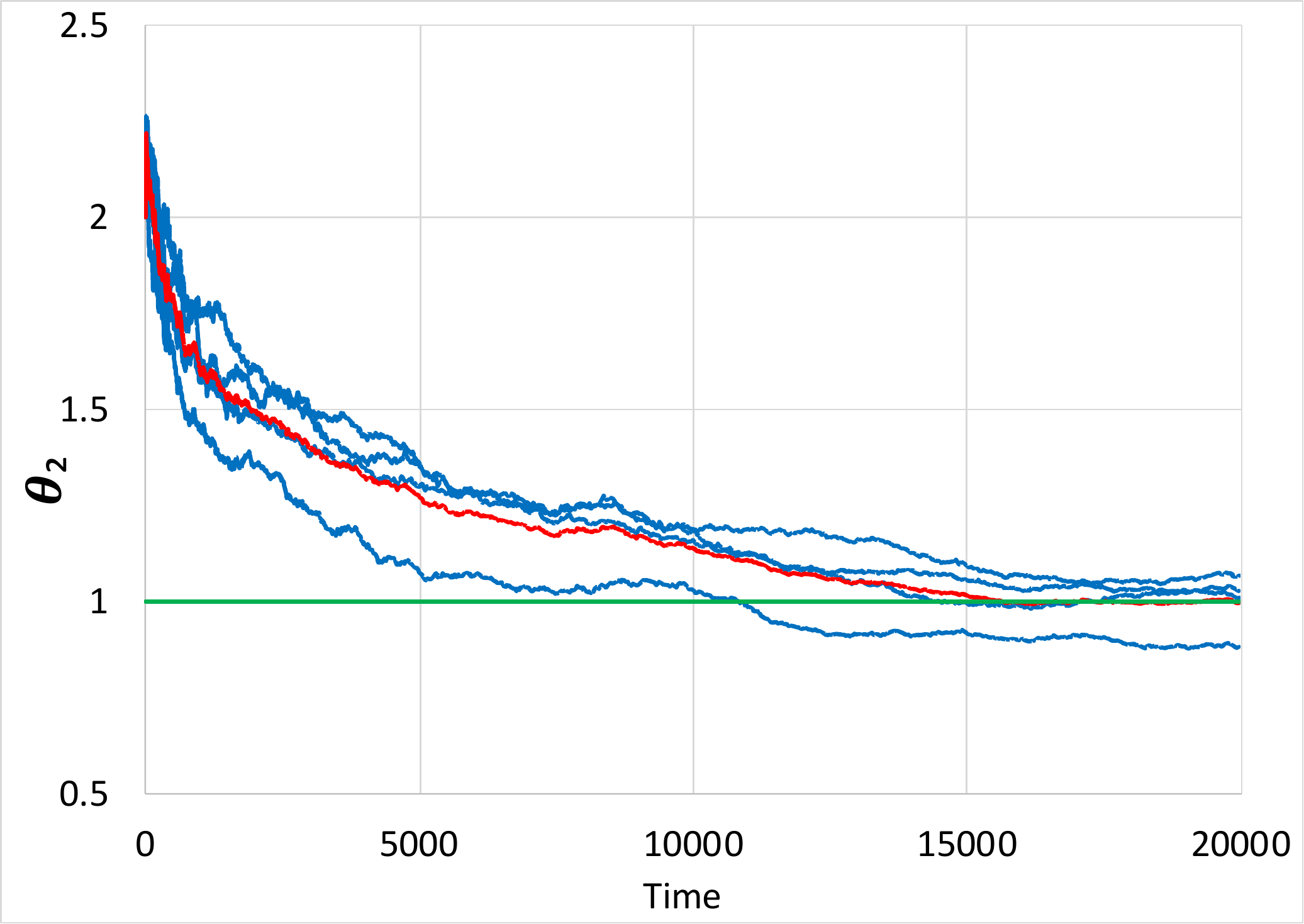} }
   \caption{The blue curves along with their average (in red) are trajectories from the execution of the algorithm in Section \ref{sec:app_par} for the estimation of $(\theta_1,\theta_2)$ in the case \textbf{F1} with $r_1=r_2=100$. The initial values of the parameters are $(-1,2)$. The green horizontal lines represent the true parameter values $(\theta_1^*,\theta_2^*)=(-2,1)$. We take $\nu_t=t^{-0.1}$ and $\kappa_t = 0.1$ when $t\leq 400$ and $\kappa_t=3 \times t^{-0.601}$ otherwise.}
    \label{fig:F1_r1=r2=100}
\end{figure}

\begin{figure} [H]
\centering
\subfloat{ \includegraphics[clip, trim=0.5cm 0.2cm 0.4cm 0.1cm,width=8.2cm, height=5.2cm]{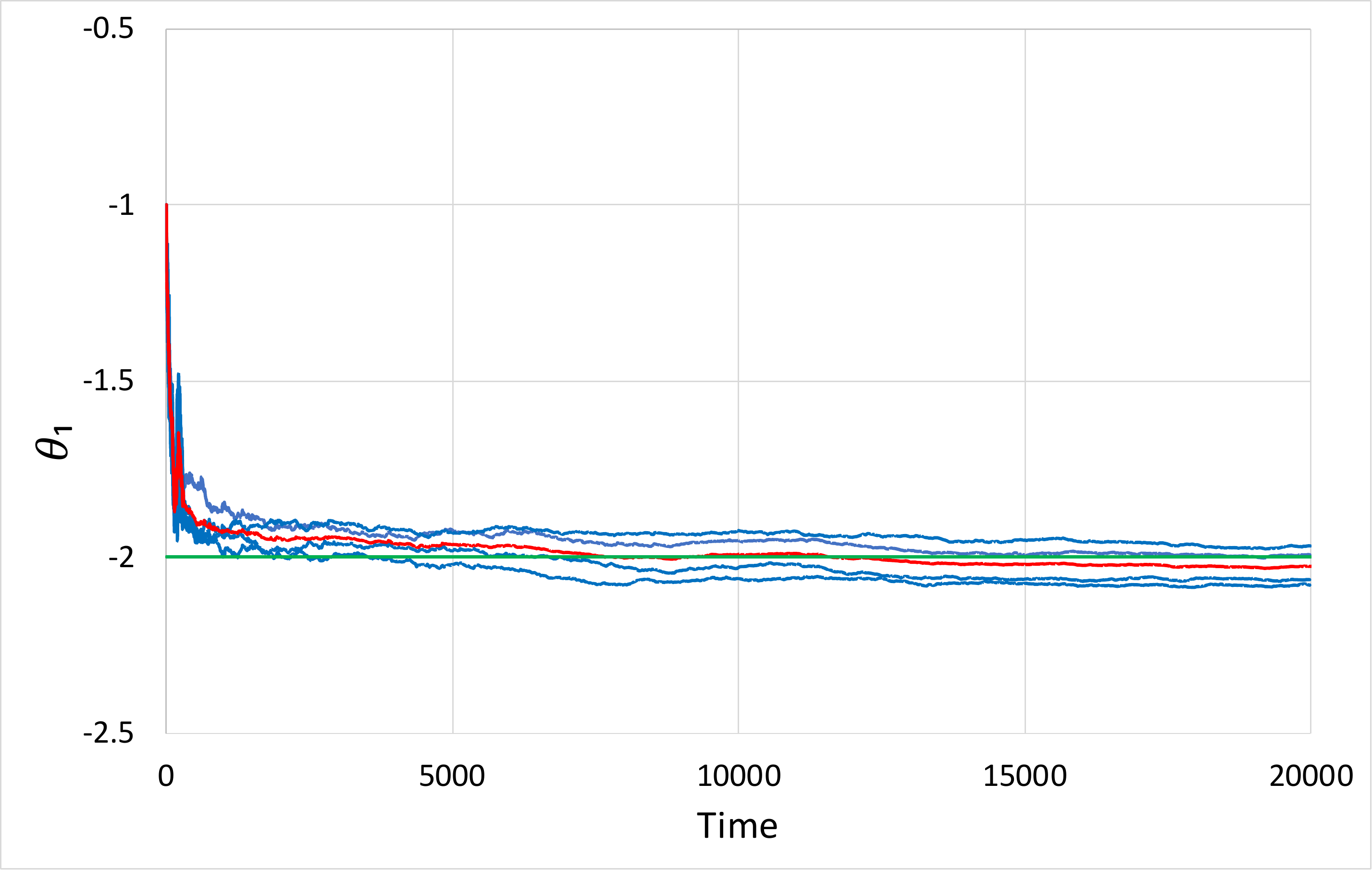} }
\subfloat{ \includegraphics[clip, trim=0.5cm 0.3cm 0.4cm 0.2cm,width=8.1cm, height=5.2cm]{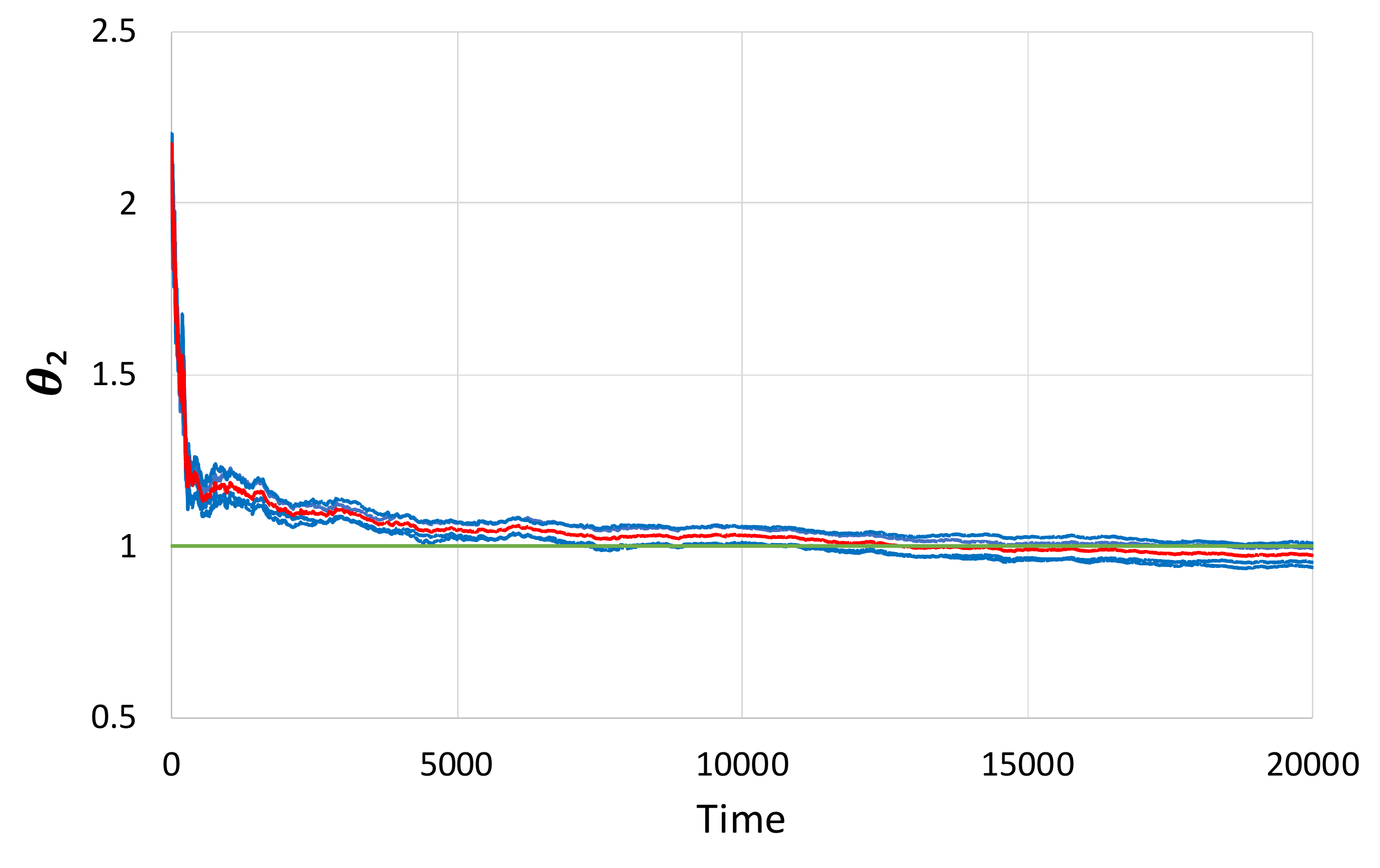} }
   \caption{The blue curves along with their average (in red) are trajectories from the execution of the algorithm in Section \ref{sec:app_par} for the estimation of $(\theta_1,\theta_2)$ in the case \textbf{F2} with $r_1=r_2=2$. The initial values of the parameters are $(-1,2)$. The green horizontal lines represent the true parameter values $(\theta_1^*,\theta_2^*)=(-2,1)$. We take $\nu_t=t^{-0.1}$ and $\kappa_t = 0.09$ when $t\leq 300$ and $\kappa_t=t^{-0.7}$ otherwise.}
    \label{fig:F2_r1=r2=2}
\end{figure}

\begin{figure} [H]
\centering
\subfloat{ \includegraphics[clip, trim=0.5cm 0.2cm 0.4cm 0.1cm,width=8cm, height=5cm]{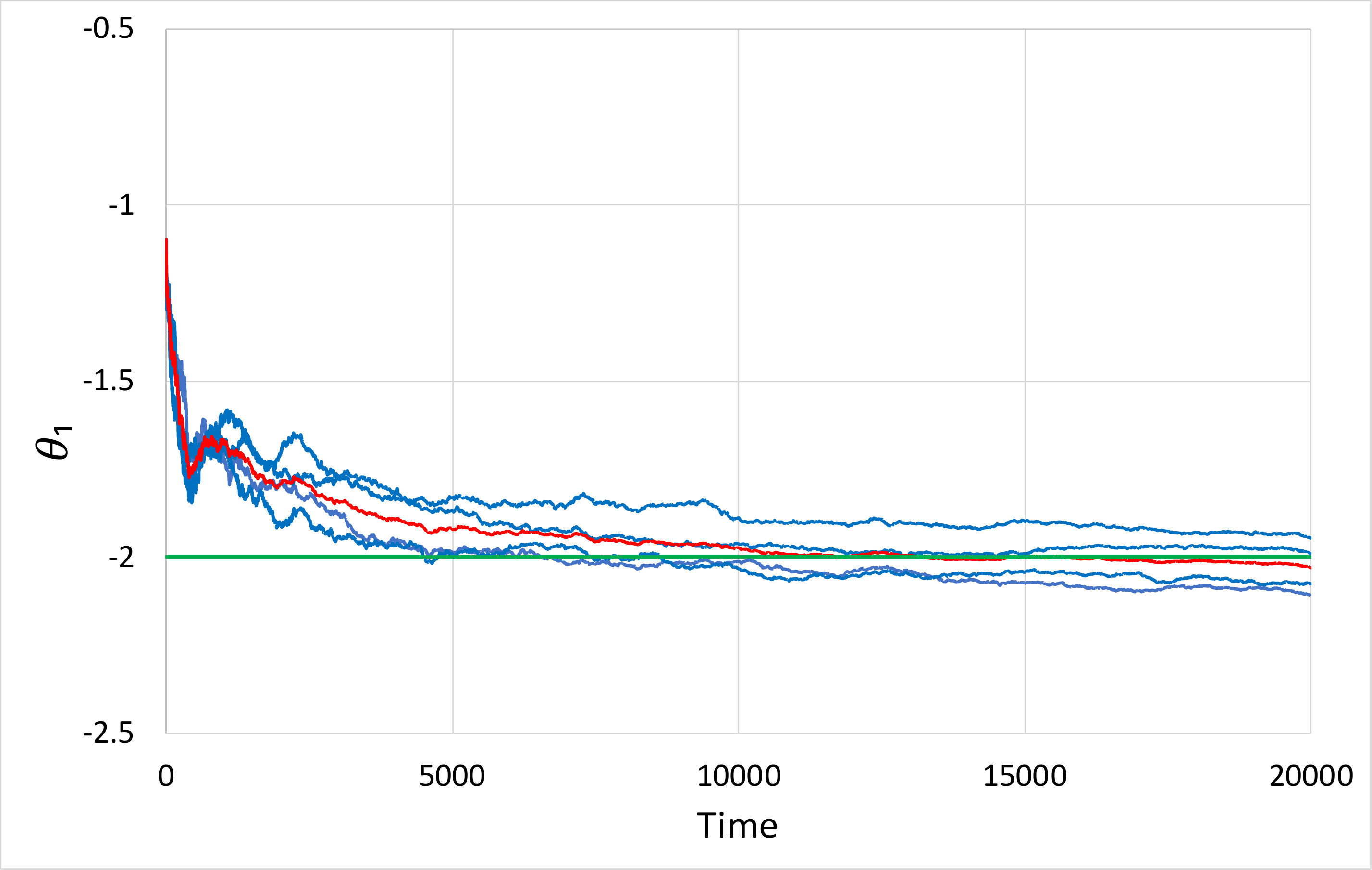} }
\subfloat{ \includegraphics[clip, trim=0.5cm 0.2cm 0.4cm 0.1cm,width=8cm, height=5cm]{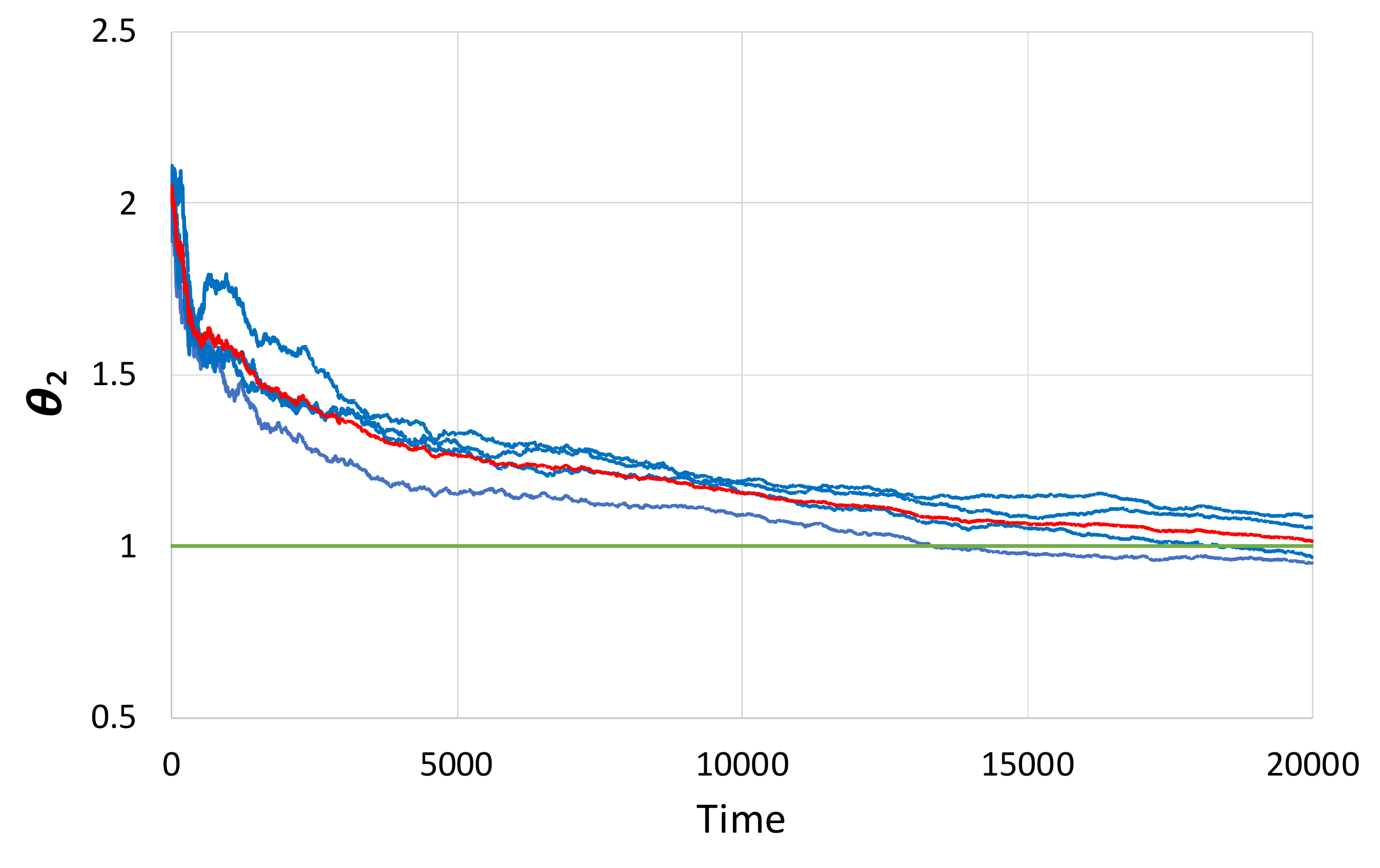} }
   \caption{The blue curves along with their average (in red) are trajectories from the execution of the algorithm in Section \ref{sec:app_par} for the estimation of $(\theta_1,\theta_2)$ in the case \textbf{F2} with $r_1=r_2=100$. The initial values of the parameters are $(-1,2)$. The green horizontal lines represent the true parameter values $(\theta_1^*,\theta_2^*)=(-2,1)$. We take $\nu_t=t^{-0.1}$ and $\kappa_t = 0.08$ when $t\leq 400$ and $\kappa_t=3 \times t^{-0.64}$ otherwise.}
    \label{fig:F2_r1=r2=100}
\end{figure}

\begin{figure} [H]
\centering
\subfloat{ \includegraphics[clip, trim=0.5cm 0.2cm 0.4cm 0.1cm,width=8cm, height=5cm]{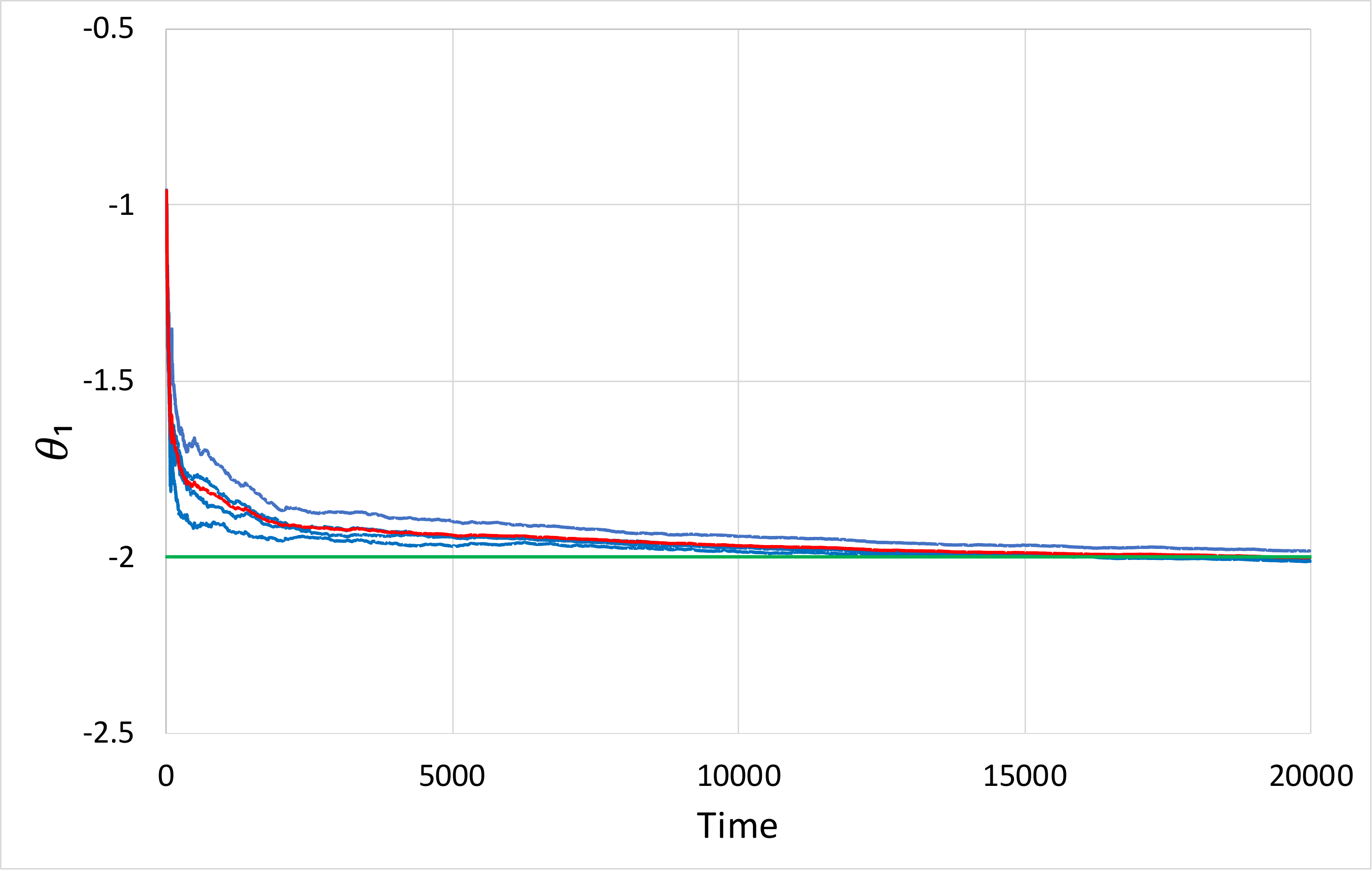} }
\subfloat{ \includegraphics[clip, trim=0.5cm 0.2cm 0.4cm 0.1cm,width=8cm, height=5cm]{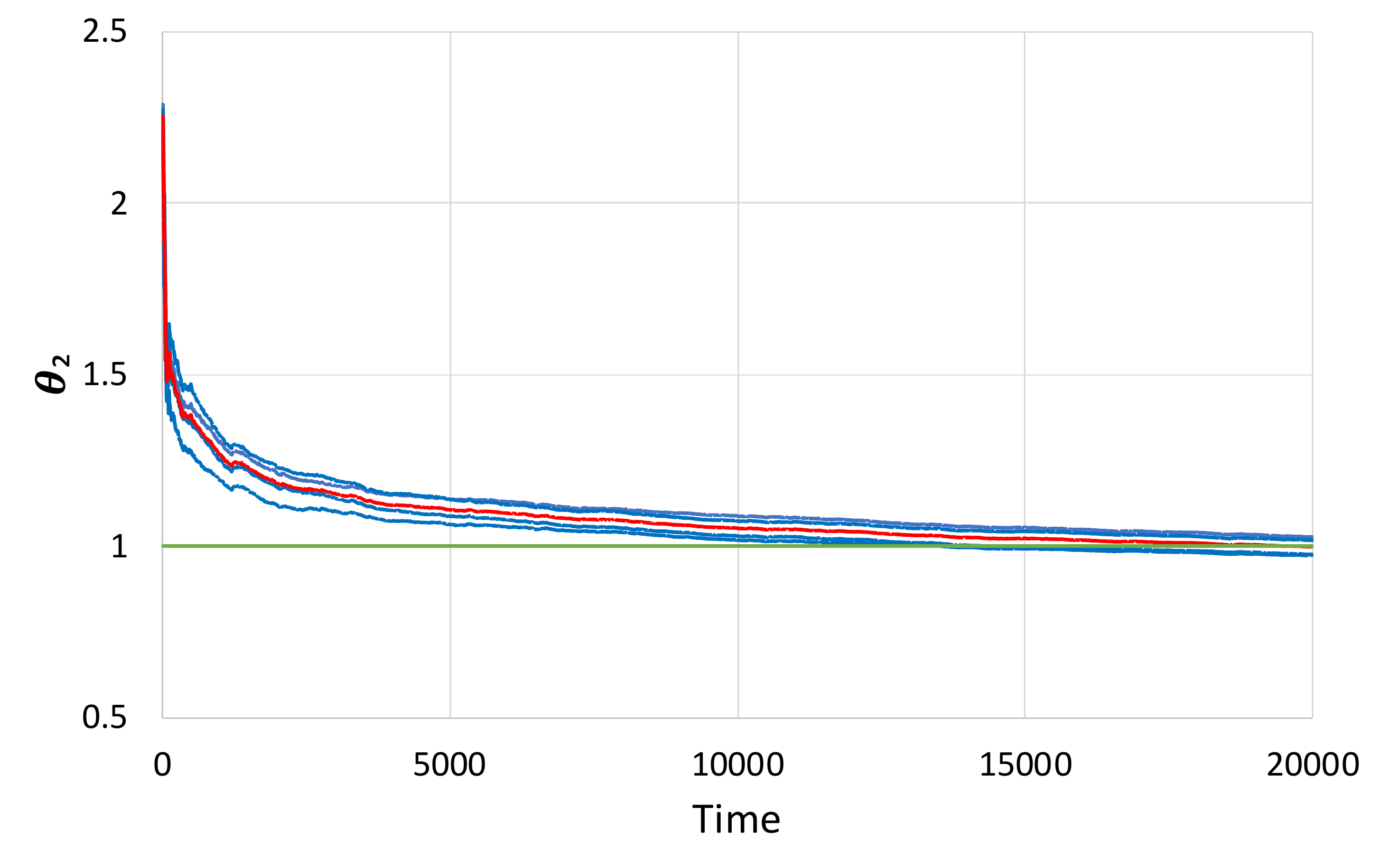} }
   \caption{The blue curves along with their average (in red) are trajectories from the execution of the algorithm in Section \ref{sec:app_par} for the estimation of $(\theta_1,\theta_2)$ in the case \textbf{F3} with $r_1=r_2=2$. The initial values of the parameters are $(-1,2.2)$. The green horizontal lines represent the true parameter values $(\theta_1^*,\theta_2^*)=(-2,1)$. We take $\nu_t=t^{-0.1}$ and $\kappa_t = 0.09$ when $t\leq 100$ and $\kappa_t= t^{-0.8}$ otherwise.}
    \label{fig:F3_r1=r2=2}
\end{figure}

\begin{figure} [H]
\centering
\subfloat{ \includegraphics[clip, trim=0.5cm 0.2cm 0.4cm 0.1cm,width=8cm, height=5cm]{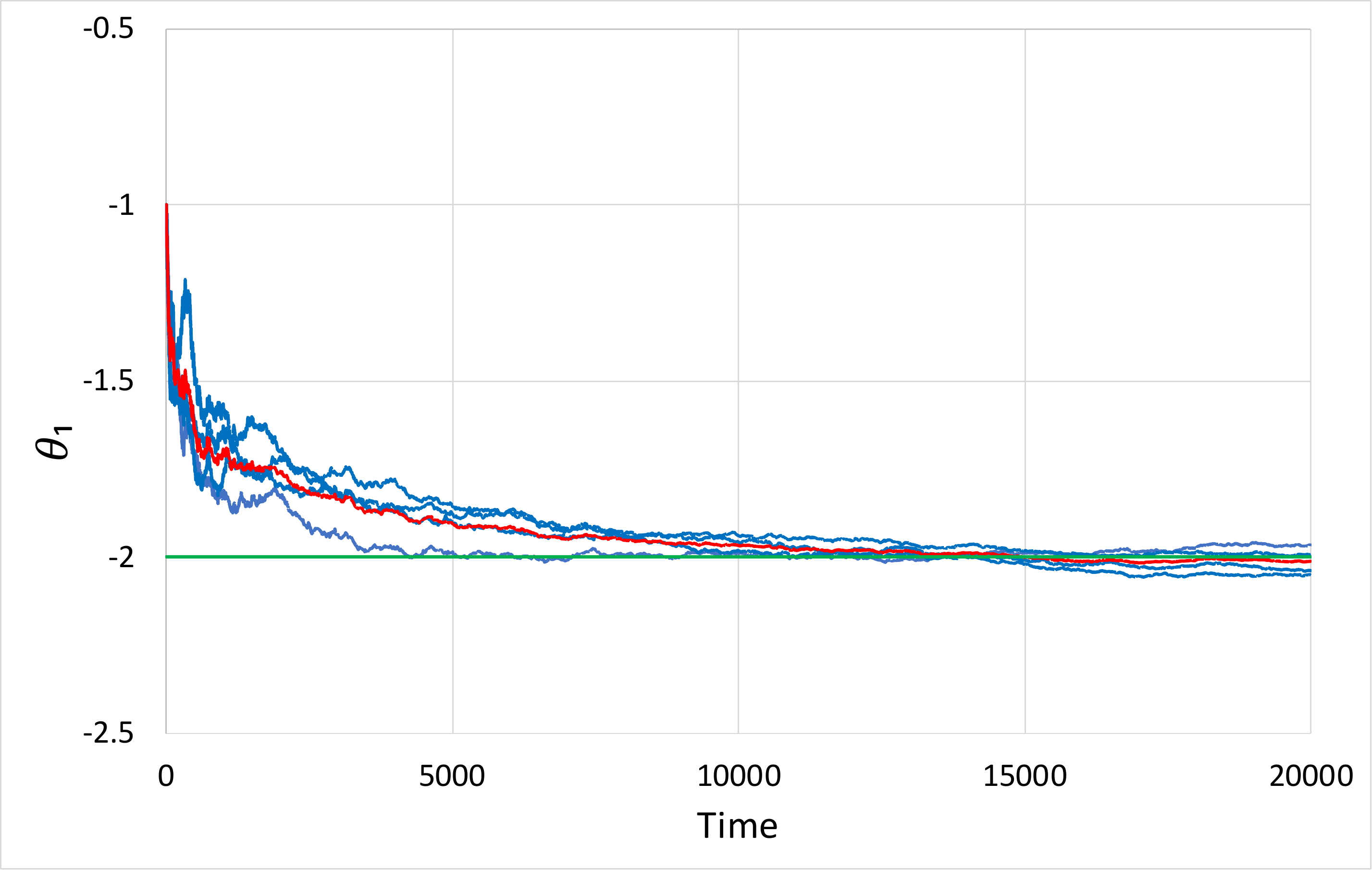} }
\subfloat{ \includegraphics[clip, trim=0.5cm 0.2cm 0.4cm 0.1cm,width=8cm, height=5cm]{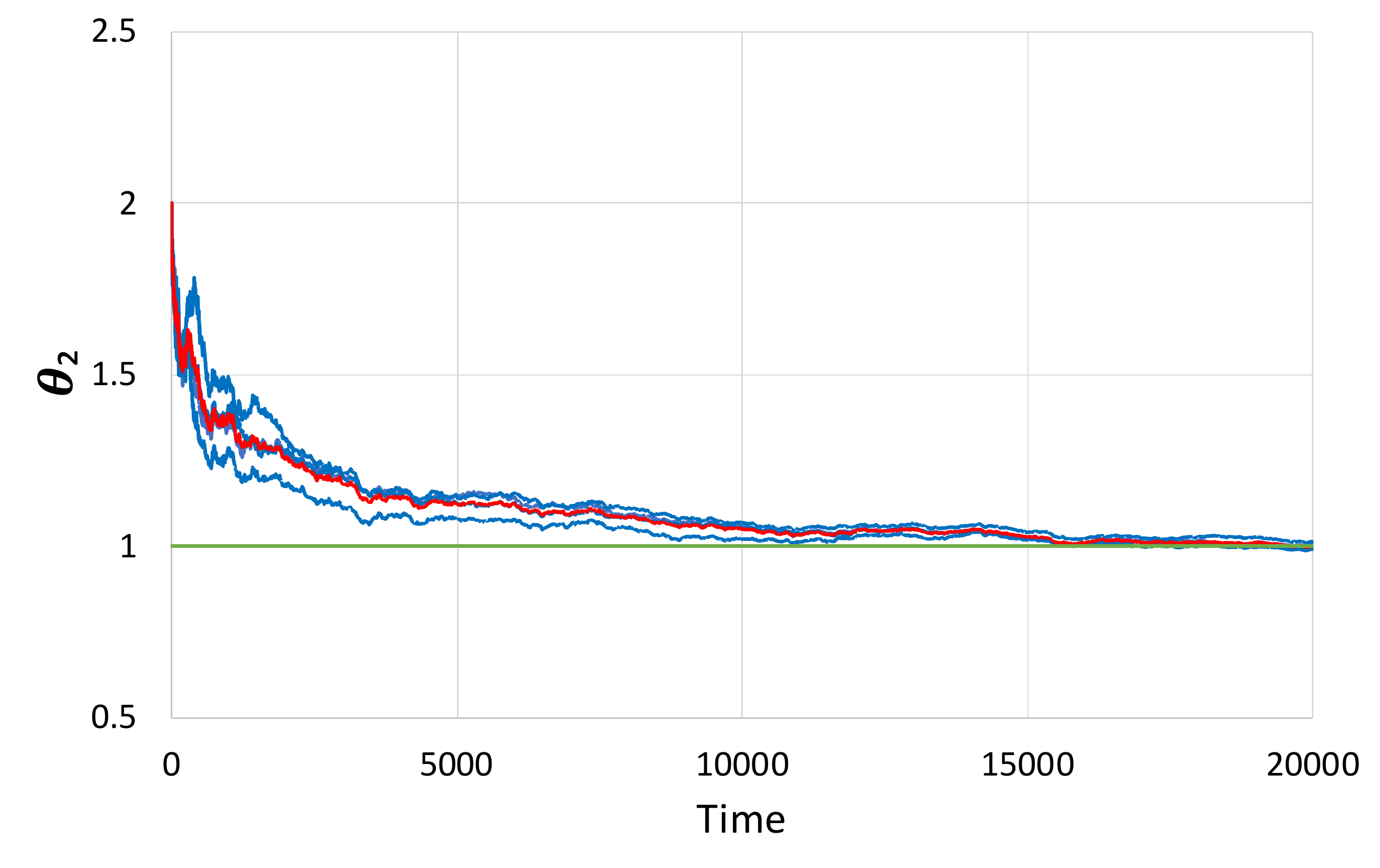} }
   \caption{The blue curves along with their average (in red) are trajectories from the execution of the algorithm in Section \ref{sec:app_par} for the estimation of $(\theta_1,\theta_2)$ in the case \textbf{F3} with $r_1=r_2=100$. The initial values of the parameters are $(-1,2)$. The green horizontal lines represent the true parameter values $(\theta_1^*,\theta_2^*)=(-2,1)$. We take $\nu_t=t^{-0.1}$ and $\kappa_t = 0.09$ when $t\leq 200$ and $\kappa_t=2 \times t^{-0.601}$ otherwise.}
    \label{fig:F3_r1=r2=100}
\end{figure}

\subsubsection{Lorenz' 63 Model}
We now consider the following nonlinear model with $r_1=r_2=3$, which is a simplified mathematical model for atmospheric convection. 
\begin{equation*}
\begin{array}{rcl}
dX_t&=&f(X_t)~dt~+~R^{1/2}_{1}~dW_t\\
dY_t&=&C~X_t~dt~+~R^{1/2}_{2}~dV_{t},
\end{array}
\end{equation*}
where 
\begin{align*}
f_1(X_t) &= \theta_1 (X_t(2) - X_t(1)),\\
f_2(X_t) &= \theta_2 X_t(1) - X_t(2) - X_t(1) X_t(3),\\
f_3(X_t) &= X_t(1) X_t(2) - \theta_3 X_t(3),
\end{align*}
where $X_t(i)$ is the $i^{th}$ component of $X_t$. We also have $R_1^{1/2}=Id$, 

\begin{align*}
\def\arraystretch{1.2}
C_{ij} = \left\{
\begin{array}{cl}
\frac{1}{2} & \text{if  } i=j\\
\frac{1}{2} & \text{if  } i=j-1\\
0 & \text{otherwise }
\end{array}
\right.,\qquad i, j \in \{1,2,3\},
\end{align*}
and $(R_2^{1/2})_{ij}=2\, q\left(\frac{2}{5}\min\{|i-j|,r_2-|i-j|\}\right)$, $ i, j \in \{1,2,3\}$, where 
\begin{align*}
\def\arraystretch{1.2}
q(x) = \left\{
\begin{array}{cl}
1-\frac{3}{2}x +\frac{1}{2}x^3 & \text{if  } 0\leq x\leq 1\\
0 & \text{otherwise }
\end{array}
\right..
\end{align*}

In Figures \ref{fig:F1_Lorenz63} - \ref{fig:F3_Lorenz63}, we show the results for the parameters estimation of $\theta=(\theta_1,\theta_2,\theta_3)$ in the $\textbf{F1}$, $\textbf{F2}$ and $\textbf{F3}$ cases. The ensemble size is $N=100$ and the discretization level is $L=8$ in all cases. The initial state is $X_0\sim \mathcal{N}(\mathbf{1}_{r_1},\frac{1}{2}Id)$, where $\mathbf{1}_{r_1}$ is a vector of 1's in $\mathbb{R}^{r_1}$.

\begin{figure} [H]
\centering
\subfloat{ \includegraphics[clip, trim=0.58cm 0.2cm 0.43cm 0.1cm,width=0.32\textwidth]{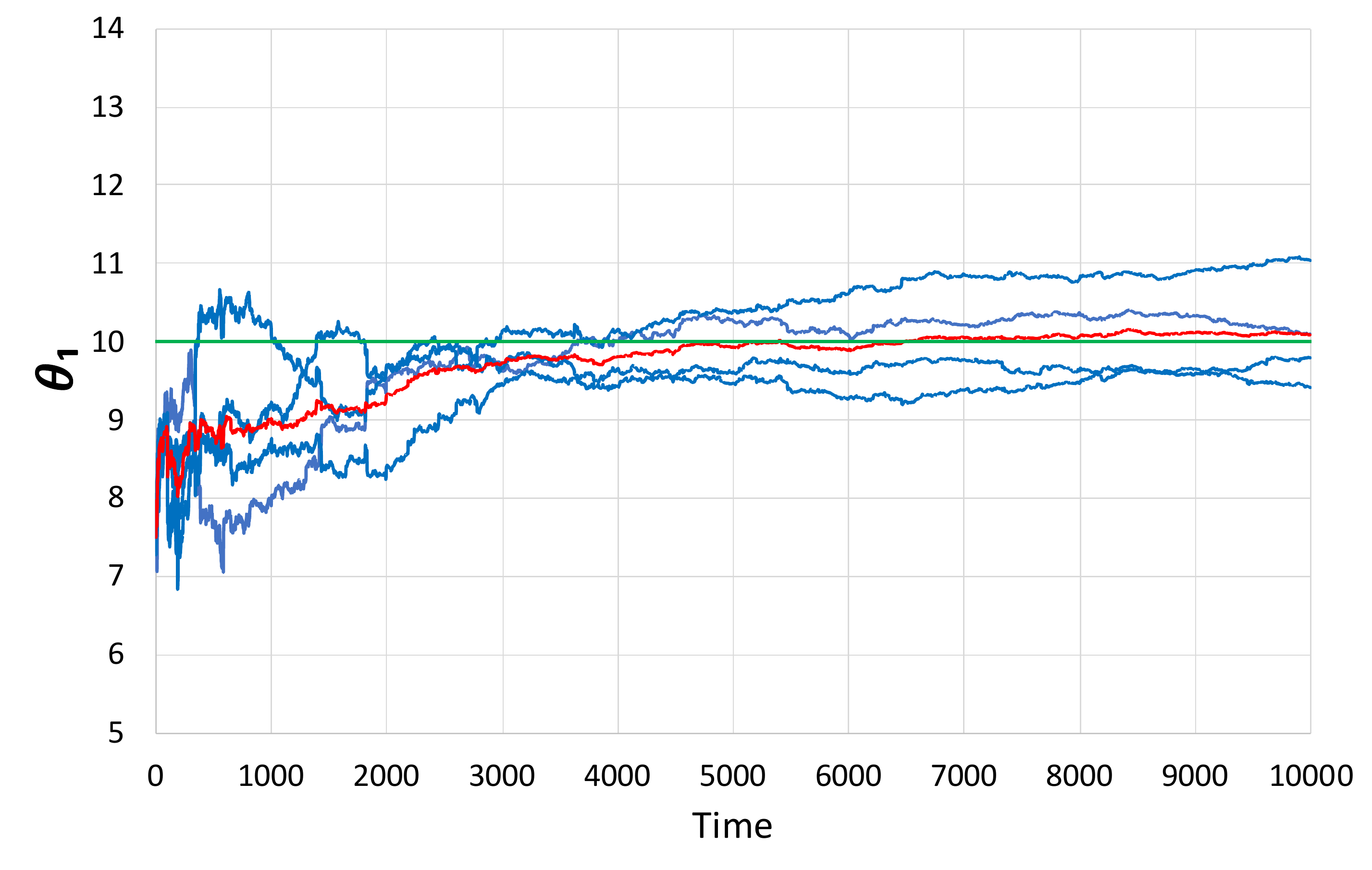} }
\subfloat{ \includegraphics[clip, trim=0.58cm 0.2cm 0.43cm 0.1cm,width=0.32\textwidth]{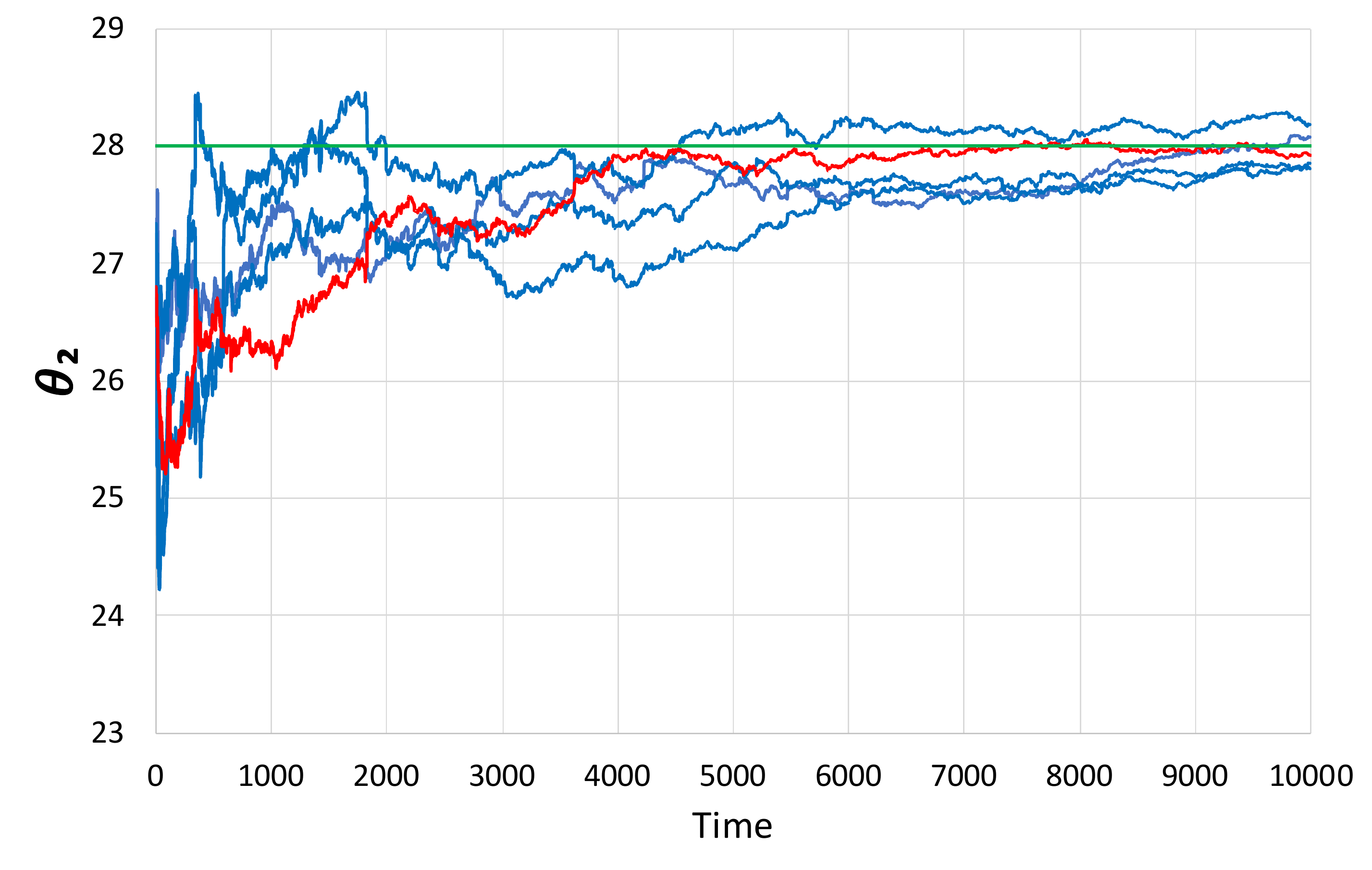} }
\subfloat{ \includegraphics[clip, trim=0.58cm 0.2cm 0.43cm 0.1cm,width=0.32\textwidth]{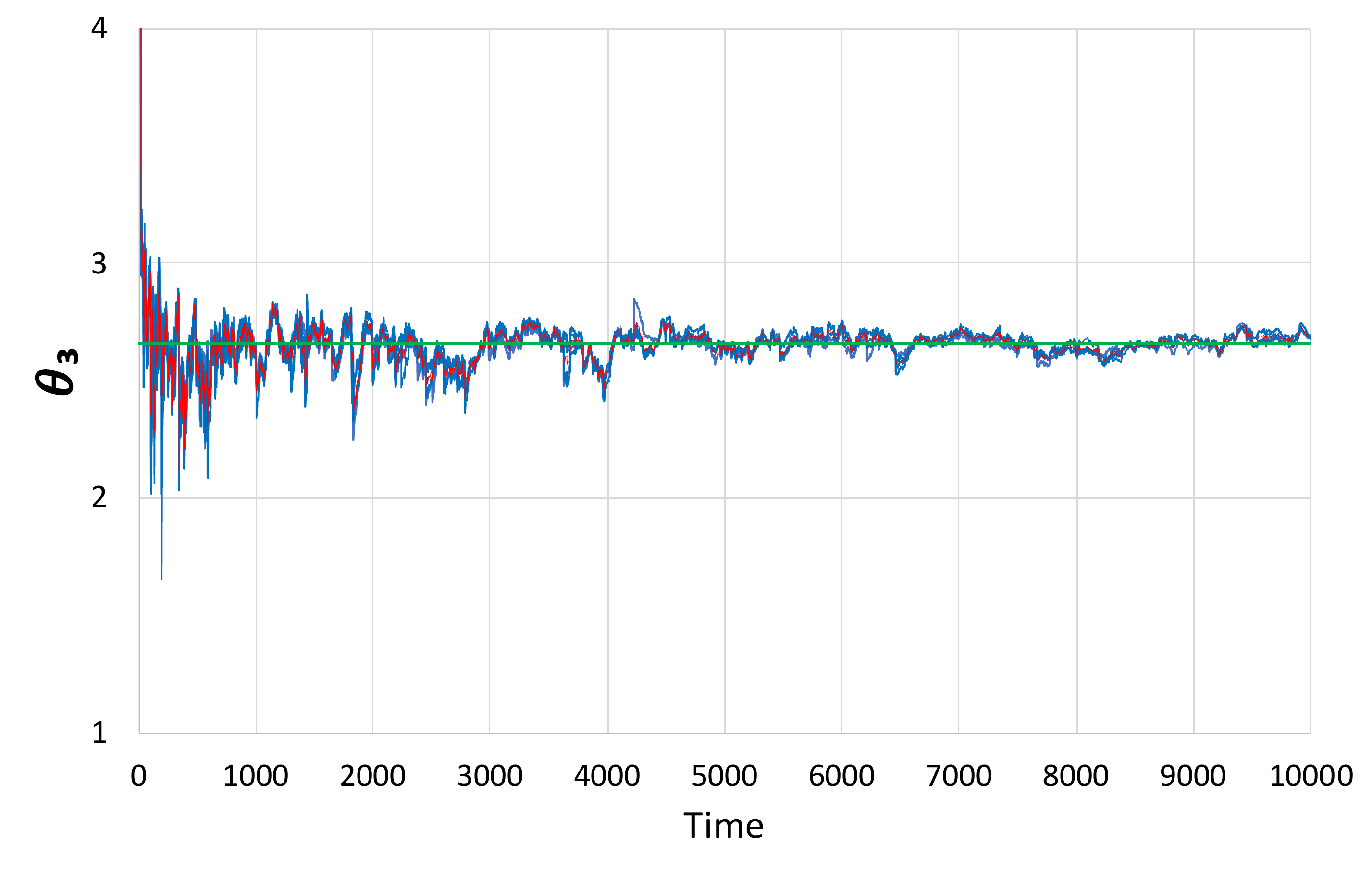} }
   \caption{The blue curves along with their average (in red) are trajectories from the execution of the algorithm in Section \ref{sec:app_par} for the estimation of $(\theta_1,\theta_2, \theta_3)$ in the case \textbf{F1}. The initial values of the parameters are $(7.5,26.7,6.5)$. The green horizontal lines represent the true parameter values $(\theta_1^*,\theta_2^*,\theta_3^*)=(10,28,\frac{8}{3})$. We take $\nu_t=t^{-0.2}$ an $\kappa_t = 0.0314$ when $t\leq 100$ and $\kappa_t= t^{-0.71}$ otherwise. }
    \label{fig:F1_Lorenz63}
\end{figure}

\begin{figure} [H]
\centering
\subfloat{ \includegraphics[clip, trim=0.58cm 0.2cm 0.43cm 0.1cm,width=0.32\textwidth]{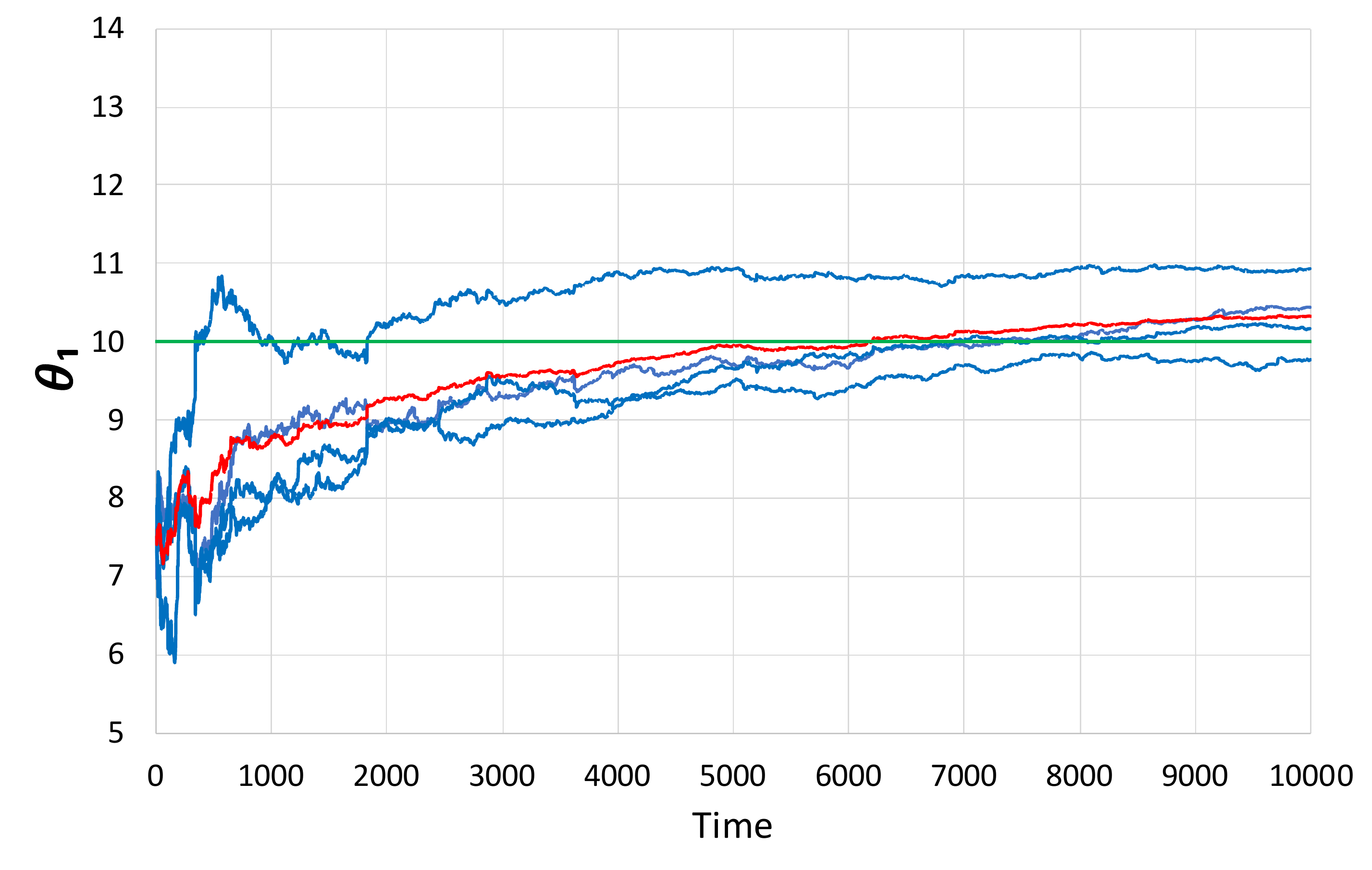} }
\subfloat{ \includegraphics[clip, trim=0.58cm 0.2cm 0.43cm 0.1cm,width=0.32\textwidth]{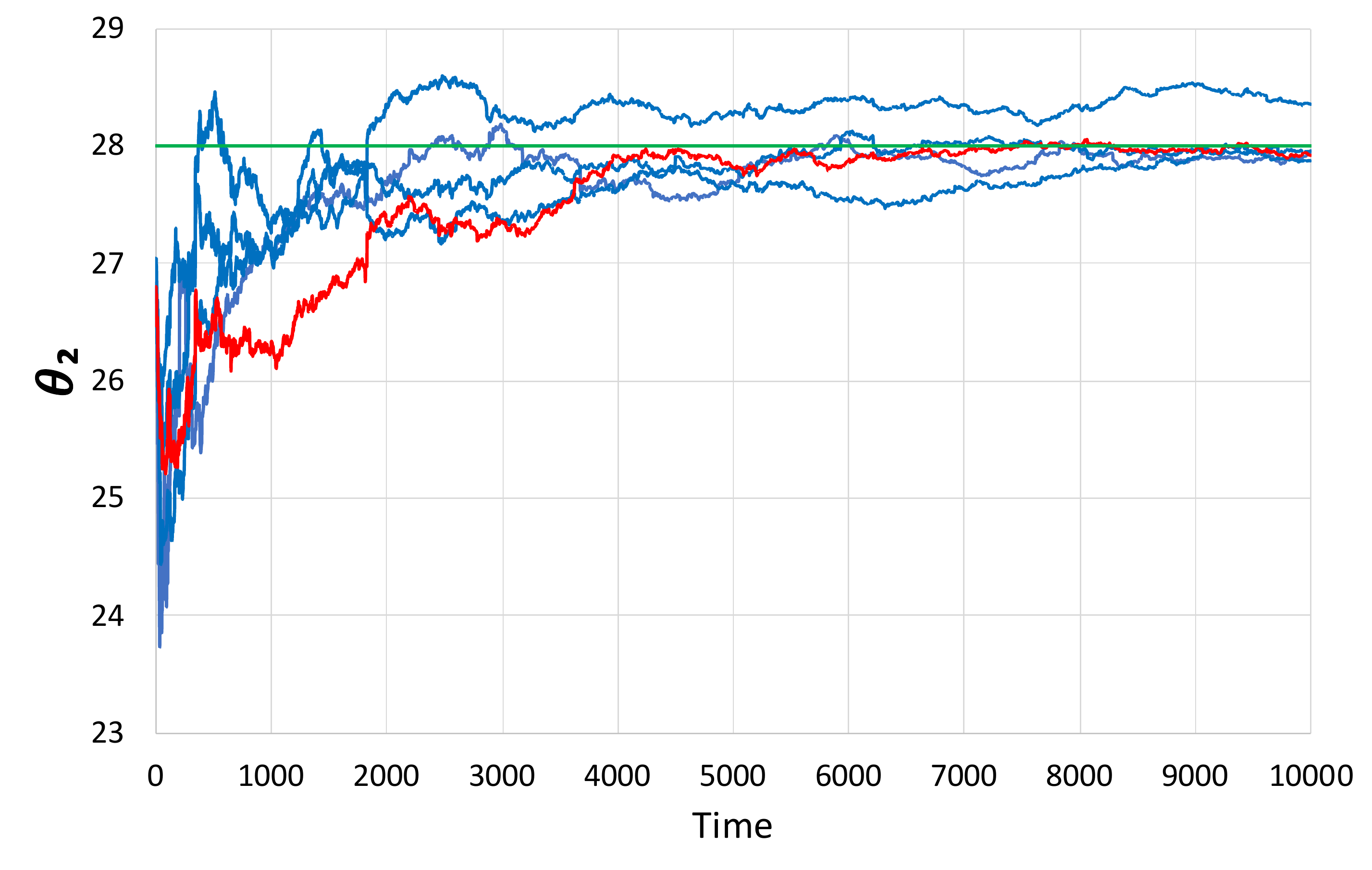} }
\subfloat{ \includegraphics[clip, trim=0.58cm 0.2cm 0.43cm 0.1cm,width=0.32\textwidth]{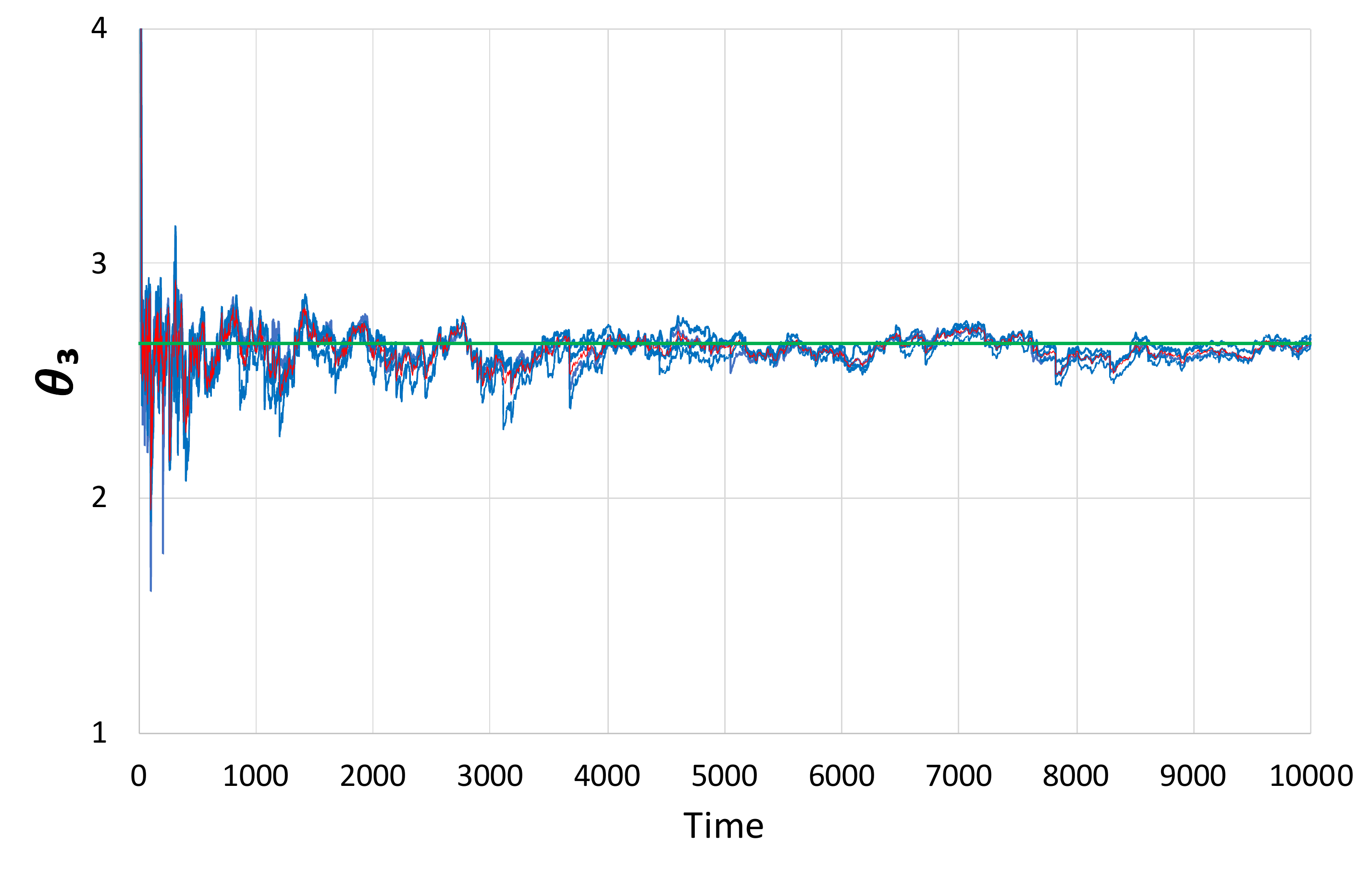} }
   \caption{The blue curves along with their average (in red) are trajectories from the execution of the algorithm in Section \ref{sec:app_par} for the estimation of $(\theta_1,\theta_2, \theta_3)$ in the case \textbf{F2}. The initial values of the parameters are $(7.5,26.7,6.5)$. The green horizontal lines represent the true parameter values $(\theta_1^*,\theta_2^*,\theta_3^*)=(10,28,\frac{8}{3})$. We take $\nu_t=t^{-0.2}$ and $\kappa_t = 0.0139$ when $t\leq 300$ and $\kappa_t= t^{-0.75}$ otherwise. }
    \label{fig:F2_Lorenz63}
\end{figure}

\begin{figure} [H]
\centering
\subfloat{ \includegraphics[clip, trim=0.58cm 0.2cm 0.43cm 0.1cm,width=0.32\textwidth]{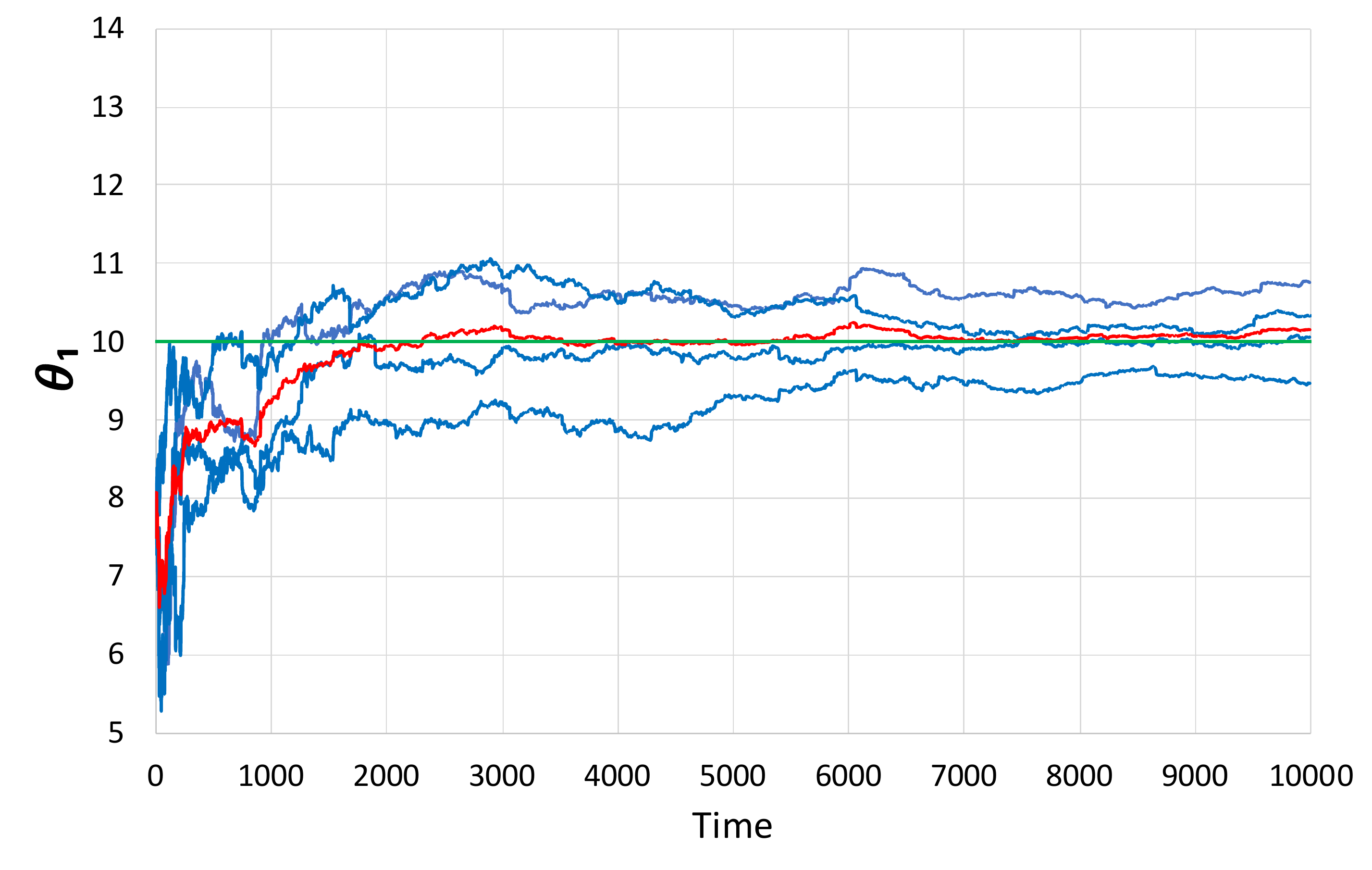} }
\subfloat{ \includegraphics[clip, trim=0.58cm 0.2cm 0.43cm 0.1cm,width=0.32\textwidth]{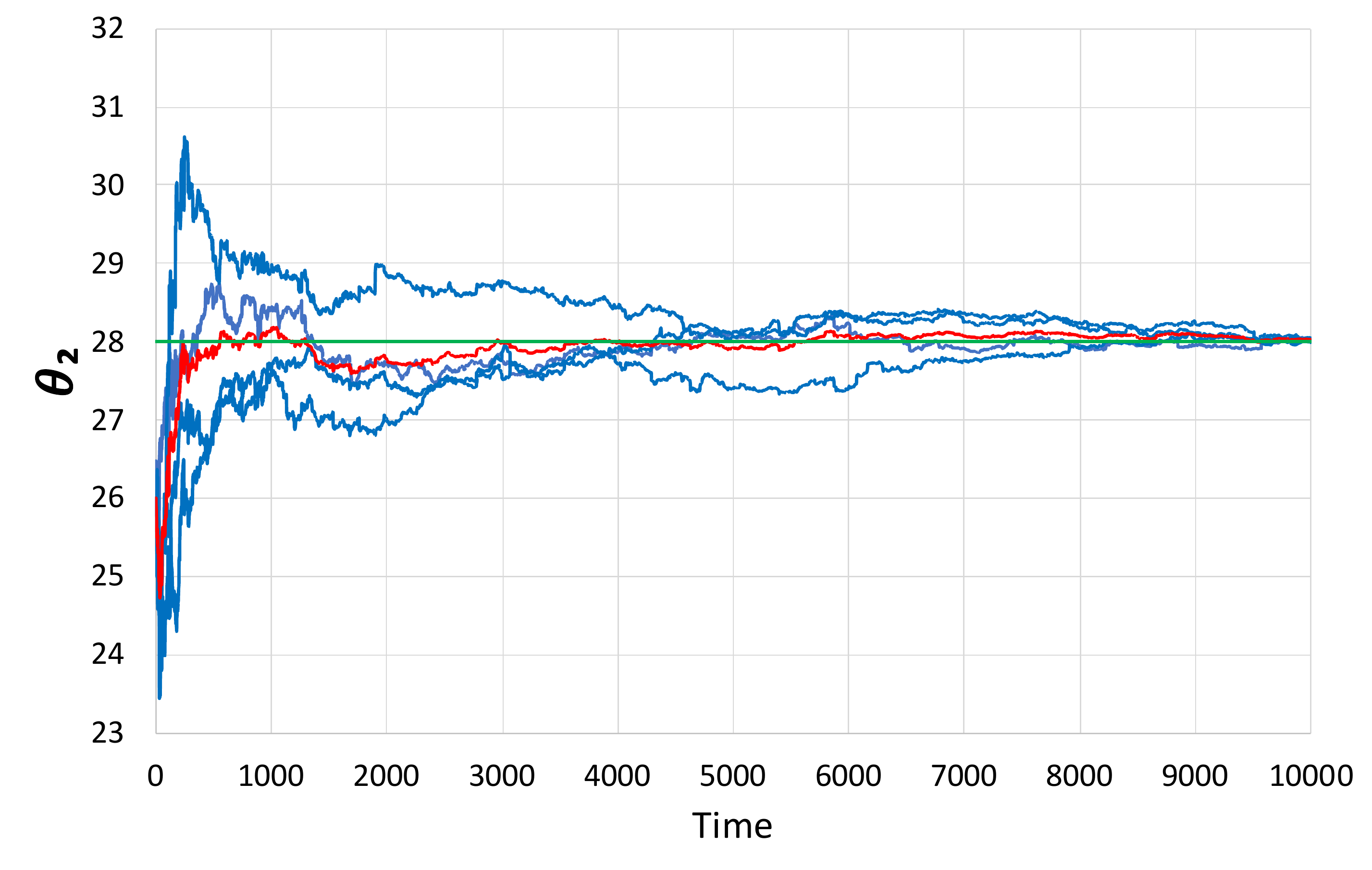} }
\subfloat{ \includegraphics[clip, trim=0.58cm 0.2cm 0.43cm 0.1cm,width=0.32\textwidth]{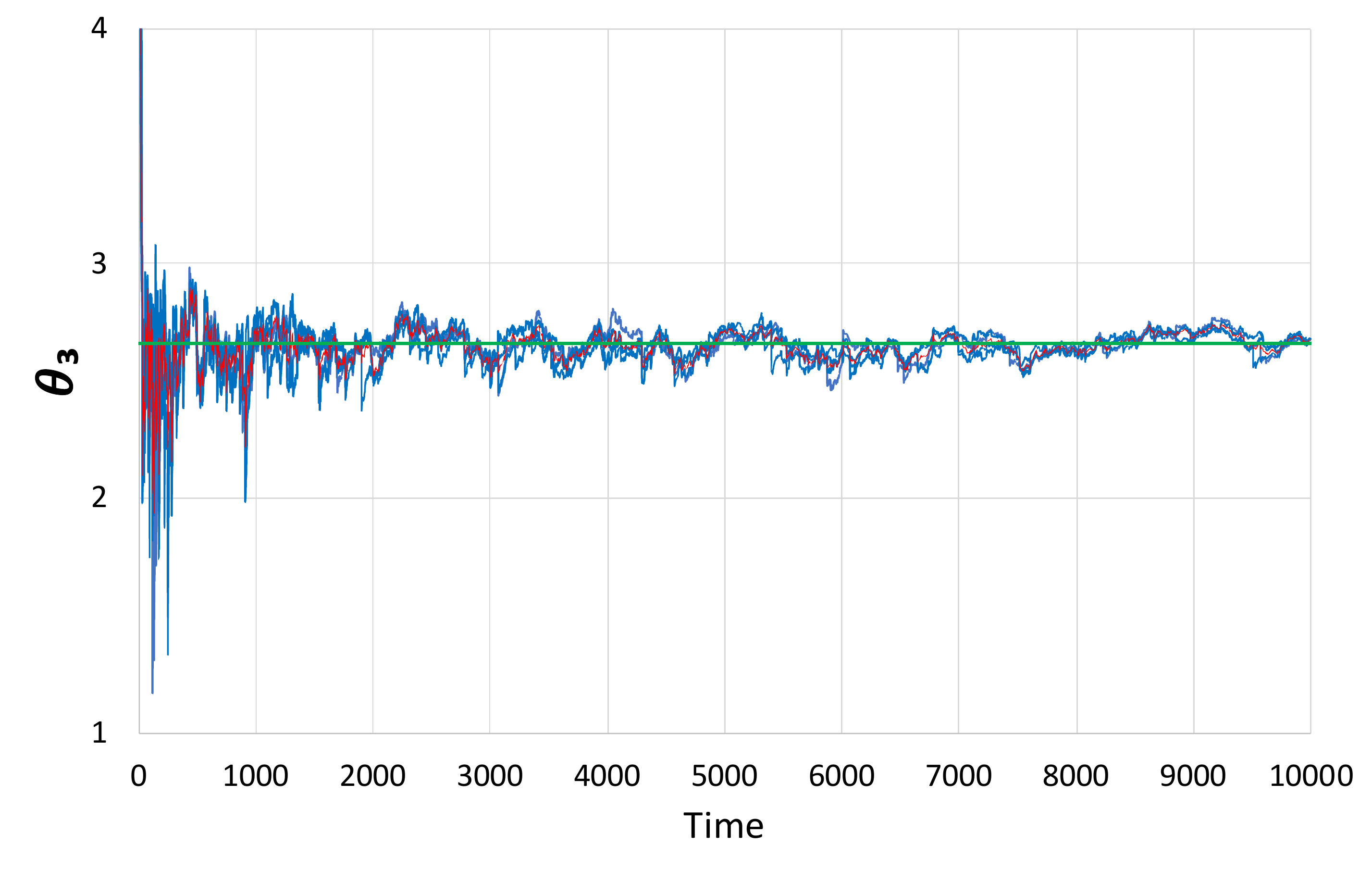} }
   \caption{The blue curves along with their average (in red) are trajectories from the execution of the algorithm in Section \ref{sec:app_par} for the estimation of $(\theta_1,\theta_2, \theta_3)$ in the case \textbf{F3}. The initial values of the parameters are $(7.5,26.7,6.5)$. The green horizontal lines represent the true parameter values $(\theta_1^*,\theta_2^*,\theta_3^*)=(10,28,\frac{8}{3})$. We take $\nu_t=t^{-0.2}$ and $\kappa_t = 0.0314$ when $t\leq 100$ and $\kappa_t= t^{-0.71}$ otherwise. }
    \label{fig:F3_Lorenz63}
\end{figure}

\subsubsection{Lorenz' 96 Model}

Finally, we consider the following nonlinear model, Lorenz' 96, with $r_1=r_2=40$. The solution of this model has a chaotic behavior and it describes the evolution of a scalar quantity on a circle of latitude. 

\begin{equation*}
\begin{array}{rcl}
dX_t&=&f(X_t)~dt~+~R^{1/2}_{1}~dW_t\\
dY_t&=&X_t~dt~+~R^{1/2}_{2}~dV_{t},
\end{array}
\end{equation*}
where 
\begin{align*}
f_i(X_t) &= (X_t(i+1) - X_t(i-2)) X_t(i-1) - X_t(i) + \theta
\end{align*}
where $X_t(i)$ is the $i^{th}$ component of $X_t$, and it is assumed that $X_t(-1)=X_t(r_1-1)$, $X_t(0)=X_t(r_1)$ and $X_t(r_1+1)=X_t(1)$.   We also have $R_1^{1/2}=\sqrt{2} Id$ and $R_2^{1/2}=\frac{1}{2} Id$. $\theta$ here represents the external force in the system, while $(X_t(i+1) - X_t(i-2)) X_t(i-1)$ is the advection-like term and $-X_t(i)$ is the damping term. In \autoref{fig:Lorenz96}, we show the results for the parameters estimation of $\theta$ in the \textbf{F1}, \textbf{F2} and \textbf{F3} cases. The ensemble size is $N=100$ and the discretization level is $L=8$ in all cases. In \textbf{F1} \& \textbf{F2} cases, $X_t$ is initialized as follows: $X_0(1) = 8.01$ and $X_0(k)=8$ for $1<k \leq 40$. In  \textbf{F3}, to avoid having the matrix $p_0^N$ equal to zero, we take $X_0 \sim \mathcal{N}(8\mathbf{1}_{r_1},0.05 Id)$.

\begin{figure} [H]
\centering
\subfloat{ \includegraphics[clip, trim=0.76cm 0.3cm 0.43cm 0.2cm,width=0.32\textwidth]{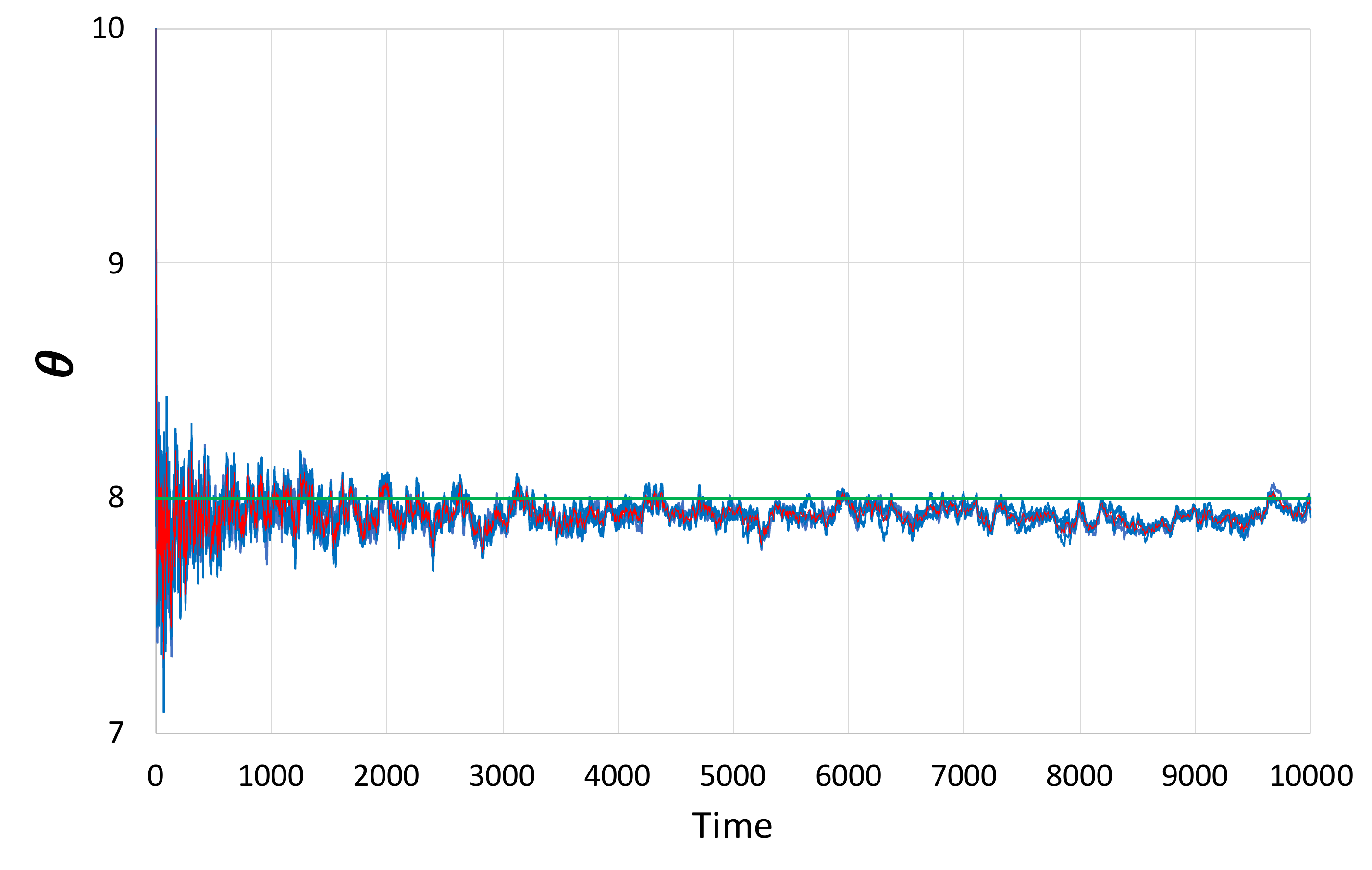} }
\subfloat{ \includegraphics[clip, trim=0.77cm 0.3cm 0.42cm 0.2cm,width=0.32\textwidth]{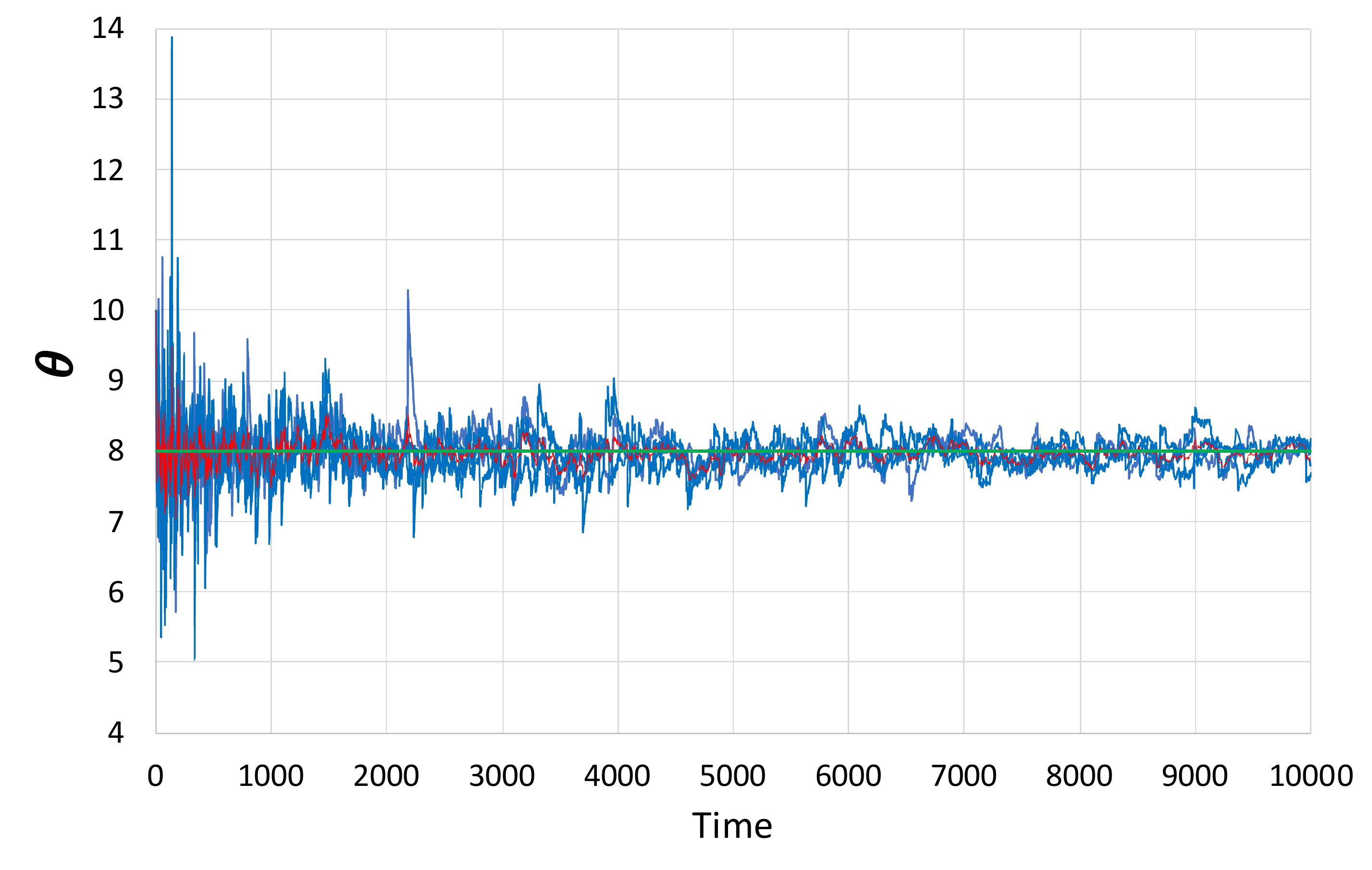} }
\subfloat{ \includegraphics[clip, trim=0.76cm 0.3cm 0.43cm 0.2cm,width=0.32\textwidth]{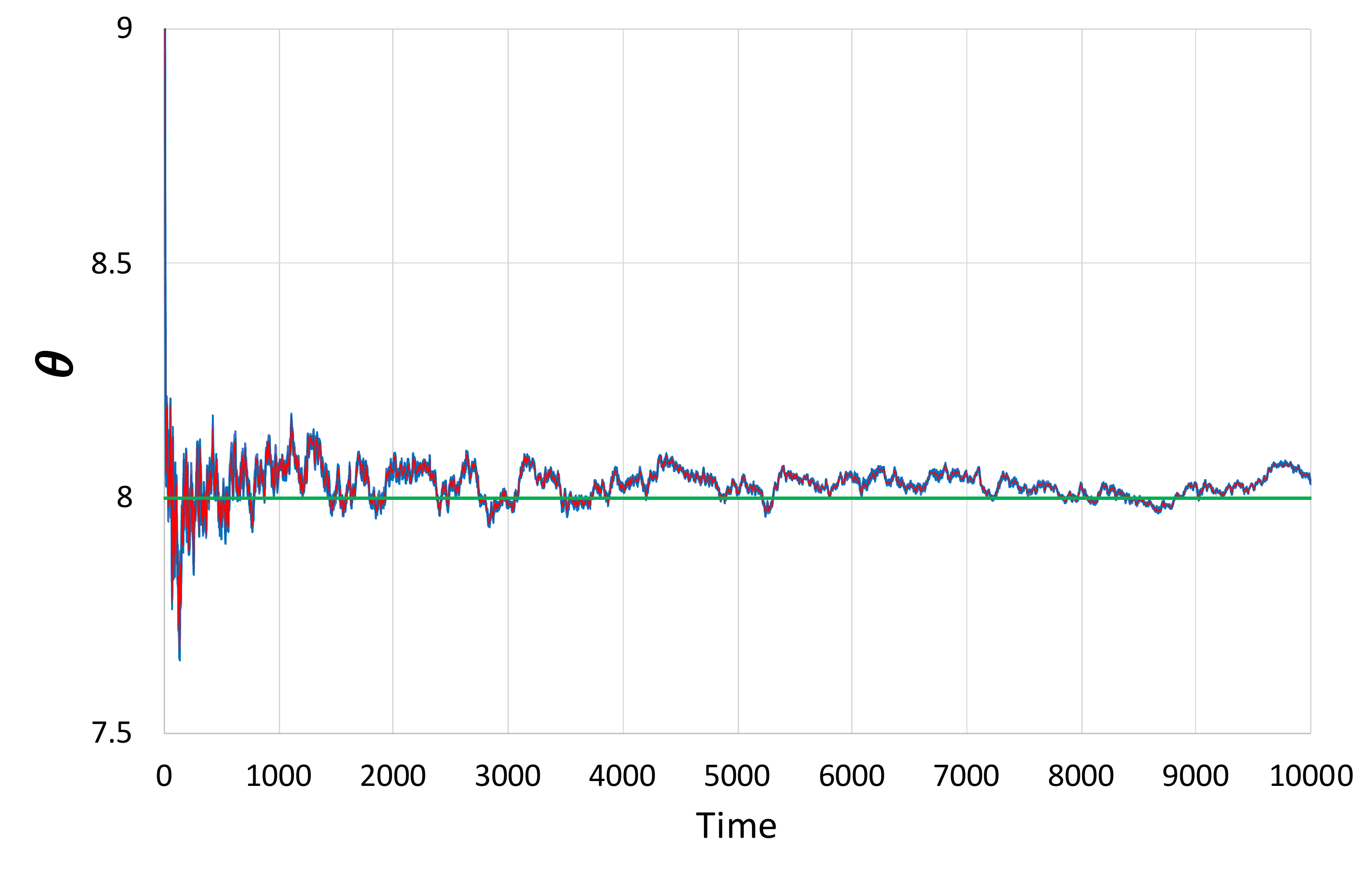} }
   \caption{The blue curves along with their average (in red) are trajectories from the execution of the algorithm in Section \ref{sec:app_par} for the estimation of $\theta$ in the case \textbf{F1} (left), \textbf{F2} (middle) and \textbf{F3} (right). The initial value of $\theta$ is 10. The green horizontal lines represent the true parameter value $\theta^*=8$. We take $\nu_t=t^{-0.1}$ and $\kappa_t = 0.0314$ when $t\leq 50$ and $\kappa_t= t^{-0.75}$ otherwise. }
    \label{fig:Lorenz96}
\end{figure}

\subsubsection*{Acknowldegements}

AJ \& HR were supported by KAUST baseline funding.
DC has been partially supported by EU project STUOD - DLV-856408.

\end{document}